\documentclass[11pt]{article}

\usepackage[english]{babel}

\usepackage{amsmath}
\usepackage{amsfonts,amsthm}
\usepackage{dsfont}
\usepackage{amssymb}
\usepackage{amsthm}
\usepackage{graphicx}
\usepackage[pdftex]{color}
\usepackage{subfigure}
\usepackage{wasysym}
\usepackage{pifont}
\usepackage{psfrag}

\newcommand{\E}{\mathbb{E}}

\newcommand{\Q}{\mathbb{Q}}

\newcommand{\R}{\mathbb{R}}

\def\ind{\mathds{1}}

\newcommand{\beq}{\begin{equation}}
\newcommand{\eeq}{\end{equation}}
\newcommand{\bi}{\begin{itemize}}
\newcommand{\bd}{\begin{description}}
\newcommand{\ei}{\end{itemize}}
\newcommand{\ed}{\end{description}}

\newcommand{\bc}{\begin{center}}
\newcommand{\ec}{\end{center}}

\newtheorem{Thm}{\bf Theorem}

\newtheorem{Prop}[Thm]{\bf Proposition}
\newtheorem{Lem}[Thm]{\bf Lemma}
\newtheorem{Ass}[Thm]{\bf Assumption}
\newtheorem{Rem}[Thm]{\bf Remark}
\newtheorem{Def}[Thm]{\bf Definition}
\newtheorem{Exa}[Thm]{\bf Example}
\newtheorem{remark}[Thm]{\bf Remark}

\numberwithin{equation}{section}
\numberwithin{Thm}{section}

\long\def\symbolfootnote[#1]#2{\begingroup%
\def\thefootnote{\fnsymbol{footnote}}\footnote[#1]{#2}\endgroup}

\title{Utility indifference pricing and hedging for structured
contracts in energy markets\footnote{This work was partly supported by the grant CPDA138873-2013 of the University of Padova ``Stochastic models with spatial structure and applications to new challenges in Mathematical Finance, with a focus on the post-2008 financial crisis environment and on energy markets''. Part of this work was done while the first and fourth authors were visiting LSE in November 2013 and while the second author was Visiting Scientist at the University of Padova in June 2014 and July 2015: the financial contribution of the two institutions is kindly acknowledged. 
Moreover, the authors wish to thank Ren\'e A\"id, Enrico Edoli, Paola Mannucci and the participants to the 2014 ISEFI Conference in Paris for valuable comments.
}}

\author{Giorgia Callegaro\footnote{Department of Mathematics,
University of Padova, via Trieste 63, 35121 Padova, Italy.}\\
\and Luciano Campi\thanks{Corresponding author. Address:
London School of Economics, Department of Statistics, Columbia
House, 10 Houghton Street, London WC2A 2AE, United Kingdom.
        Email: L.Campi@lse.ac.uk}\\ 
\and Valeria Giusto\footnote{Phinergy S.r.l.s., Via della Croce Rossa 112, 35129 Padova, Italy.} \\         
\and Tiziano Vargiolu\footnotemark[2] \\
}

\date{February 20, 2016  \\
}

\definecolor{greenGio}{cmyk}{1.0,0.0,1.0,0}
\definecolor{blueGio}{cmyk}{1.0,0.0,0.0,0}
\definecolor{greyGio}{rgb}{0.5,0.5,0.5}

\begin{document}

\maketitle

\begin{abstract}
In this paper we study the pricing and hedging of structured products in
energy markets, such as swing and virtual gas storage, using the exponential utility indifference pricing approach in a general incomplete multivariate market model driven by finitely many stochastic factors. The buyer of such contracts is allowed to trade in the forward market in order to hedge the risk of his position.
We fully characterize the buyer's utility indifference price of a given product in terms of continuous viscosity solutions of
suitable nonlinear PDEs.
This gives a way to identify reasonable candidates for the optimal
exercise strategy for the structured product as well as for the corresponding hedging strategy.
Moreover, in a model with two correlated assets, one traded and one nontraded, we
obtain a representation of the price as the value function of an
auxiliary simpler optimization problem under a risk neutral
probability, that can be viewed as a perturbation of the minimal
entropy martingale measure.
Finally, numerical results are provided.
\medskip



\end{abstract}

\section{Introduction}
In the last fifteen years, since the start of the energy market deregulation and
privatization in Europe and in the U.S., the study of energy
markets became a challenging topic from both a practical and a theoretical perspective. Especially important is the problem of pricing and hedging of energy contracts. This is far from being trivial because of the peculiarity of the models and since these contracts typically have a very complex structure, incorporating optionality features which can be exercised by the buyer at multiple times. The two main examples of products used in energy markets for primary supply are swing contracts and forward contracts.
While the structure of the forward contracts is rather simple, swing contracts are much more involved as they give the buyer some degrees of freedom about the
quantity of energy to buy for each sub-period, usually with daily
or monthly scale, subject to a cumulated constraint over the
contract period. This flexibility is much welcomed by the contract
buyers, as energy markets are affected by many unexpected events
such as peaks in consumption related to sudden weather changes,
breakdowns of power plants, financial turmoils and so on.
Many other kinds of
contract are traded in the energy market, they are often negotiated
over-the-counter, and some of them, e.g. virtual storage
contracts, also include optionality components similar to the ones of swing contracts.

The pricing of these products has
a consolidated tradition in discrete time models (see, e.g. \cite{EFRV} or \cite{HeLaRusso2013} and references therein), which is manly based on dynamic programming. The two papers \cite{PaBaBo09} and \cite{PaBrWil10} propose a different method based on optimal quantization theory. In continuous time models, the first approaches were based on optimal switching techniques (e.g., \cite{CL10}) or multiple stopping (e.g., \cite{CT08}). In all these articles the optionality features of the structured product can be exercised over a discrete set of stopping times, that can be chosen by the buyer.

A different approach consists in approximating the contract payoff with its continuous time counterpart. This idea has been proposed in \cite{BLN11} (and further exploited in \cite{BCV13}) for swing
contracts and in \cite{ChenForsyth06,Felix12,TDR09} for virtual
storage contracts. Other examples of structured contracts can be treated with the same methodology, see e.g. Benth and Eriksson \cite{BE} for flexible load contracts and tolling agreements. The main advantage of this approach is that it makes the pricing problem more tractable, since it allows using the stochastic control theory in continuous time, based on PDE methods. In those papers, the price of a structured
contract is defined - in analogy with American options - as the supremum, over all the strategies available to the buyer,
of the expected payoff, where the expectation is taken under a given risk-neutral measure. When the model for the underlying of the structured contract is Markovian, as it happens in most of the models used in practice, the pricing problem
reduces to solving the corresponding Hamilton-Jacobi-Bellman (HJB) equation. Notice that the choice of the risk-neutral measure is not at all obvious since energy market models are typically incomplete, because of the presence of assets bearing non tradable risks. Moreover, with the exception of \cite{warin} which focuses on gas storage contracts and uses a delta-hedging
approach, the problem of hedging the risk coming from a long position in those structured products is not considered in those papers. Hedging such contracts can be quite a delicate task in energy markets, since the underlying of the contracts is often not tradable, hence the buyer has to trade in some other asset with a good correlation with the underlying.  For an extensive review of the existing literature with descriptions of the most traded contracts and a detailed comparison between the main articles we refer to the very recent book \cite{Aidbook}. \medskip

Our contribution to the literature consists in building on the idea of continuously approximating the payoff as in  \cite{BCV13,BLN11}, in order to provide a general framework, where both problems of pricing and hedging of structured contracts can be solved in a consistent fashion. The main novelty of this paper is that the buyer of the given structured contract is allowed to (at least partially) hedge his position by trading in forward contracts, written on the underlying of the structured contract itself or on some asset correlated with the underlying. We model the forward market as a general incomplete
multivariate market model with finitely many forward contracts (with different maturities), evolving over time as diffusions whose coefficients depend on a certain number of exogenous stochastic factors with Markovian dynamics. The underlying of the structured contract is defined as a function of such factors. This setting includes many models that have been previously proposed and studied in the literature, e.g. \cite{ACLP,CL06,CV08,SS00} to cite a few.

Being the market incomplete, we adopt the utility indifference price (henceforth UIP) approach, which is one of the most appealing ways of pricing in incomplete markets, since it naturally incorporates the buyer preferences in the price of the contract. We assume that the preferences of the buyer can be encoded in an exponential utility function with a risk aversion parameter $\gamma >0$. The UIP approach has been extensively used for pricing European and American options in a wide range of financial market models. We refer to \cite{HendersonHobson09} for an excellent survey on this approach. This approach was already used for energy derivatives in Fiorenzani \cite{Fiorenz}, but to our knowledge never for the evaluation of structured products.

We apply this method for evaluating a rather general structured derivative. Its buying UIP will be characterized as the difference between the two log-value functions of the agent (with and without the contract), that can be obtained as the unique viscosity solutions of a suitable HJB equations. Our results are consistent with the ones in
\cite{BCV13,BLN11,ChenForsyth06,Felix12,TDR09}, in the case of complete market models. Moreover, the shape of such HJB equation gives reasonable candidates for the optimal withdrawal strategy of the
structured product, as well as for the related hedging strategy.

Finally, we push our general results further in two specific examples. One of them is a class of models with two risky assets, one traded and one nontraded, and constant
correlation. This includes models with a nontraded asset and basis risk, which have been studied by many authors (see, e.g., the papers \cite{Davis06,Henderson02,Monoyios04} to cite only a few). For these models, we obtain a representation of the price as the value
function of an auxiliary simpler optimization problem under a
risk-neutral probability, that can be viewed as a perturbation of
the minimal entropy martingale measure. Such a perturbation is due to the dependence of drift and volatility of the traded asset on the nontraded one and depends on the value function without the contract. It seems that such a measure change is new to the incomplete market literature. The second example is based on a slight generalization of the two-factor model developed in \cite{CV08} for energy markets, where the factors can be correlated.
\medskip

The paper is organized as follows. In Section \ref{sec:theProblem} we formulate the
problem of pricing, by introducing the general payoff of the structured
contracts, the market model and the (exponential) utility
indifference price. In Section \ref{sec:UIPvisco}, we characterize the UIP in terms
of viscosity solutions of suitable HJB equations. In Section \ref{ex} we consider the two examples described above while Section \ref{sec:num} presents some
numerical applications of our results. Finally, Section \ref{conclusion} concludes.\\

{\noindent \textbf{Notation.}
In what follows, unless explicitly stated, vectors will be column
vectors, the symbol ``*'' will denote transposition and the trace
of a square matrix $A$ will be denoted by $\textrm{tr}(A)$.
Furthermore, $\langle a, b \rangle := a^* b$ will stand for the
Euclidean scalar product. We choose as matricial norm $| A | =
\sqrt{{\rm{tr}}(AA^*)}$. On the set $\mathcal S_n$ of all
symmetric squared matrices of order $n$, we define the order
$A\le B$ if and only if $B-A \in \mathcal S_n ^+$, the subspace
of nonnegative definite matrices in $\mathcal S_n$. We will
denote by $I_n$ the identity matrix of dimension $n$.

\section{Formulation of the problem}\label{sec:theProblem}

Let $T>0$ be a finite time horizon. All the processes introduced below
will be defined on the canonical probability space $(\Omega, \mathcal F,
\mathbb P)$, where $\Omega := C([0,T];\mathbb R^d)$ is the space of all continuous functions from $[0,T]$ into $\mathbb R^d$. For $\omega \in \Omega$, we set $W_t (\omega) = \omega(t)$ and define $(\mathcal
F_t)_{t\in [0,T]}$ as the smallest right-continuous filtration such that $W$ is optional. Moreover, $\mathcal F := \mathcal F_T$. We let $\mathbb P$ be the Wiener measure on $(\Omega, \mathbb F_T)$. We can assume without loss of generality that such a filtration is complete.

\subsection{Structured products}
In this section we give a short description of the main structured products that are traded in energy markets. The typical payoff is given by a family of random variables
\begin{equation} \label{eq:C_T def}
C_T^u :=\int_0^T L(P_s,Z_s^u,u_s) ds + \Phi(P_T,Z_T^u),
\end{equation}
indexed by a control $u$, which typically represents the marginal
quantity of commodity purchased and it belongs to a suitable set of
admissible controls $\mathcal U$ that we will specify later. In particular, the admissible controls take values in some bounded interval $[0,\bar u]$ for a given threshold $\bar u >0$.
The variable $P$ in the equation (\ref{eq:C_T def}) above denotes the spot price of
the commodity (e.g., gas) and $Z^u _t := z_0 + \int_0^t u_s ds$
for all $t \in [0,T]$, for some initial value $z_0
\ge 0$. For technical reasons, that will become clear in the proofs of our results, we will need the following assumption on the structured products.

\begin{Ass}\label{AssC} 
The functions $L: \mathbb R \times [0,\bar u T] \times
[0,\bar u] \to \mathbb R$ and $\Phi : \mathbb R \times [0,\bar u
T] \to \mathbb R$ in (\ref{eq:C_T def}) are continuous and bounded. 
\end{Ass} 

The most common structured products in energy markets are swing
and virtual storage contracts. More details are given just below. See also the subsequent Remark \ref{remC} explaining how one can safely modify these contracts in order to satisfy Assumption \ref{AssC}.

\begin{Exa}[Swing contract]\label{ex:swing} {\rm
For a swing contract one has (see, e.g., \cite{BCV13,BLN11})
$$ L(p,z,u) = u (p - K), $$
where $K$ is the purchase price or strike price, and $u$ is any admissible control. These products usually include some additional features,
such as inter-temporal constraints on $u$ or on the cumulated
control $Z^u$ or some penalty function appearing in the payoff.
More precisely, constraints on $u$ and $Z^u$ are typically of the
form $Z_T^u \in [m,M]$, with $0 \leq m < M$,
with possibly further intermediate constraints on $Z_{t_i}^u$,
$t_i < T$, $i = 1,\ldots,k$. In the absence of such additional
constraints, a penalty is usually present which can be expressed
as a function $\Phi$ of the terminal spot price $P_T$ and
cumulated consumption $Z^u_T$. A typical form of $\Phi$ is
\begin{equation}\label{eq:Phi}
\Phi(p,z) = - C \left( (m - z)^+ + (z - M)^+\right) 
\end{equation}
for constants $C>0$ and $0 \leq m < M$(see \cite{BCV13,BLN11} and
references therein). We will focus on the latter case, i.e., a
non-zero penalty function $\Phi (P_T ,Z_T^u)$ without any other
contraints on the admissible controls.
} \end{Exa}

\begin{Exa}[Virtual storage contract]\label{ex:virtualSt} {\rm
These products replicate a physical gas storage position, while
being handled as pure trading contracts. In this case one has
$$ L(p,z,u) = p (u - a(z,u)), \qquad \Phi(p,z) = - C (M - z), $$
with $C,M > 0$ suitable constants, $a(z,u) := \bar a \ind_{u<0}$
and where the control $u$ is such that
$$ u_t \in [u_{\mathrm{in}}(Z_t^u),u_{\mathrm{out}}(Z_t^u)]
,\quad t\in [0,T], $$
where $u_{\mathrm{in}},u_{\mathrm{out}}$ are suitable
deterministic functions given by the physics of fluids: their
typical shapes are
$$ u_{\mathrm{in}}(z) := - K_1 \sqrt{\frac{1}{z + Z_b} + K_2},
\qquad u_{\mathrm{out}}(z) := K_3 \sqrt{z} $$
with $Z_b,K_i > 0$, $i = 1,2,3$, given constants
\cite{ChenForsyth06,Felix12,TDR09}.} \end{Exa}

\begin{Rem}\label{remC}
{\rm The boundedness of $L$ as in Assumption \ref{AssC} is not verified in the two Examples \ref{ex:swing} and
\ref{ex:virtualSt}, where $L$ is linear in $p$, which can in
principle take any real value. In practice, one can artificially
bound $L$, for example by introducing
$$
\tilde L(p,z,u) := \max(- \kappa, \min(L(p,z,u),\kappa)),
$$
so that $|\tilde L(p,z,u)| \leq \kappa$ for all $(p,z,u)$, for a
suitably chosen and large enough threshold $\kappa >0$ such that
the instantaneous profit should not be larger than $\kappa$ in
absolute value with high probability. The same truncation argument can be applied to
the penalty function $\Phi(p,z)$. Alternatively, one could truncate the unbounded variable $p$ appearing in both payoffs $L$ and $\Phi$.} 
\end{Rem}

\subsection{The market model}\label{model}

The spot price of the commodity, $P$, underlying the structured products, is modelled as 
$P_t := p(t,X_t)$, where $p:[0,T] \times \mathbb R^m \to \mathbb R$ is a measurable function and $X$ represents the factors driving the market. We assume that the process $X$ has Markovian dynamics
\begin{equation}\label{dynX}
dX_t = b(t,X_t)\ dt + \Sigma^*(t,X_t)\ dW_t, \quad X_0 =x \in
\mathbb R^m,
\end{equation}
with drift $b:[0,T] \times \mathbb R^m
\rightarrow \mathbb R^m$ and volatility matrix $\Sigma:[0,T] \times \mathbb R^m
\rightarrow \mathbb R^{d \times m}$. \medskip

We also assume that $n\leq d$ forward contracts on the commodity $P$ are traded in the
market, with maturities $T_1 < \ldots < T_n$, with $T_1 \geq T$.
Letting $F^i$ to denote the price of the forward contract with maturity
$T_i$, $i = 1,\ldots,n$, we assume that the dynamics of $F :=
(F^1,\ldots,F^n)$ is given by
\begin{equation}\label{eq:marketDynIncompletemulti}
dF_t = \mathrm{diag}(F_t) \left( \mu_F (t,X_t) dt + \sigma_F ^* (t,X_t)
dW_t \right), \quad F_0 = f_0 \in \mathbb R^n,
\end{equation}
for some functions $\mu_F : [0,T]\times \mathbb R^m \to \mathbb R^n$ and
$\sigma_F: [0,T]\times \mathbb R^m \to \mathbb R^{d\times n}$. 

Assumptions on the coefficients of $X$ and $F$ are given below. We will always assume throughout the paper that the interest rate is zero.\medskip

Notice that the forward contracts are not necessarily written on
the commodity with spot price $P$, as they could also be written
on some correlated commodity. For instance, $P$ could be the spot price
of gasoline, while the $F$'s are written on oil, as in
\cite{CL06,Fiorenz}. This can be also due to illiquidity or
to the fact that forward contracts relative to the commodity do not exist: for
a detailed discussion of this phenomenon, see \cite[Section
2.3]{CL06}.\medskip

We will always work under the following standing assumptions on the coefficients of the model:
\begin{Ass}\label{ass-coeff}
\begin{enumerate}
\item[(i)] The function $p: [0,T] \times \mathbb R^m \to \mathbb
R$ is continuous.
\item[(ii)] The coefficients $b:[0,T] \times \mathbb R^m
\rightarrow \mathbb R^m$ and $\Sigma:[0,T] \times \mathbb R^m
\rightarrow \mathbb R^{d \times m}$ of the factor process $X$ are continuous functions, Lipschitz in $x$ uniformly in $t$ and with linear growth in $x$ uniformly in $t$.
\item[(iii)] The drift $\mu_F : [0,T]\times \mathbb R^m \to \mathbb R^n$ and the volatility $\sigma_F: [0,T]\times \mathbb R^m \to \mathbb R^{d\times n}$ are continuous functions, Lipschitz in $x$ uniformly in $t$ and with linear growth in $x$ uniformly in $t$. \end{enumerate}
\end{Ass}
Under such assumptions, the SDEs
(\ref{dynX}) and (\ref{eq:marketDynIncompletemulti}) are
well-known to admit a unique strong solution $(X,F)$ such that
$X_0 =x$ and $F_0 = f_0$ (see, e.g., Theorem 13.1 in \cite[Chapter V]{RW2}).

\subsection{Admissible strategies and utility indifference price}

We consider an agent whose preferences are modelled by an exponential utility
function $U(x)=-\frac{1}{\gamma}e^{-\gamma x}$ , $x \in
\mathbb R$, with risk aversion parameter $\gamma > 0$. We assume that (s)he has a long position in $q \geq 0$ units of a given structured product
with payoff $C_T = (C_T
^u)_{u\in \mathcal U}$ with $C_T ^u $ as in (\ref{eq:C_T def}). Moreover, in order to hedge away the risk attached to such a contract, (s)he trades in the financial
market of forward contracts described in the previous section. At any time $s \in [0,T]$, the agent
invests the amount of wealth $\pi^i_s$ in the
forward contracts $F^i$ with $i = 1,\ldots,n$. Hence the evolution of the agent's portfolio is
$$
\left\langle \pi_s, \frac{dF_s}{F_s} \right\rangle = \sum_{i=1}^n
\pi^i_s \frac{dF^i_s}{F^i_s} = \langle \pi_s, \mu_F (s,X_s) ds +
\sigma_F^* (s,X_s) dW_s \rangle,
$$
where we recall that $\langle \cdot, \cdot \rangle$ denotes the Euclidean scalar product in $\R^n$ and we use the notation
$$ \frac{dF_s}{F_s} :=
\left(\frac{dF^i_s}{F^i_s}\right)_{i=1,\ldots,n} = \mu_F (s,X_s)
ds + \sigma_F^* (s,X_s) dW_s , \quad s\in [0,T] .$$
At this point, we need to specify the set $\mathcal A$ of admissible strategies.

\begin{Def} \label{Def:2.6}
Let $\bar u >0$ be a given threshold. The set of admissible
controls $\mathcal A$ is the set of all couples $(u,\pi)$, where
$u$ is any adapted process such that $u_t \in [0,\bar u]$ for all
$t\in [0,T]$, and $\pi$ is any progressively measurable $\mathbb
R^n$-valued process such that
\begin{equation}
\sup_{t\in [0,T]} \mathbb E\left [\exp\left (\varepsilon | \pi_t |^2 \right)\right] < \infty , \label{adm}
\end{equation}
for some $\varepsilon >0$.
We will denote by $\mathcal U$ the set of all admissible controls
$u$. Moreover, $\mathcal A_t$ (resp. $\mathcal U_t$) will be the
set of admissible controls $(u,\pi)$ (resp. admissible controls
$u$) starting from $t$.
\end{Def}

Now, we are ready to introduce the utility indifference
(buying) price of $q$ units of the structured product $C_T$. We will use the notation $C_{t,T}^u$ for the payoff
of the structured contract $C_T^u$ starting at time $t$, i.e.,
$$
C_{t,T}^u =  \int_t^T L(P_s,Z_s^u,u_s) ds + \Phi(P_T,Z_T^u), \quad t\in [0,T].
$$
Moreover, we set $C_T^u = C_{0,T}^u$.

\begin{Def}
The utility indifference (buying) price at time $t$ for a
position $q \ge 0$ in the structured product, when starting from
the initial portfolio value $y_t$, is defined as the unique
$\mathcal F_t$-measurable random variable $v_t $ solution (whenever it exists) to
\begin{equation} \label{UIP}
V(y_t - v_t,q) = V(y_t,0),
\end{equation}
where
\begin{equation}\label{eq:V}
V(y_t,q):= \sup_{(u,\pi) \in \mathcal A_t} \mathbb E_t \left[
-\frac{1}{\gamma}\exp \left( - \gamma \left( y_t + \int_t^T
\left\langle \pi_s, \frac{d F_s}{F_s} \right\rangle + q C_{t,T}^u
\right) \right) \right],
\end{equation}
and $\mathbb E_t$ stands for the conditional expectation given
$\mathcal F_t$. \end{Def}

Clearly, $V(y_0,q)$ gives the maximal expected utility from
terminal wealth, computed at time $0$, that an agent with an
exponential utility can obtain starting from an initial wealth
$y_0$ and having a position $q \ge 0$ in the structured product. Therefore, the (buying) UIP defined above represents the highest price the buyer is willing to pay for $q$ units of the structured contract.\medskip

The maximization problem (\ref{eq:V}) can be easily translated into a standard Markovian control problem by suitably redefining the set of state variables as follows.
Let $t \in [0,T]$. First, using Equation (\ref{eq:C_T def}), we can rewrite the terminal wealth as follows
$$ y_t + \int_t^T \left\langle \pi_s, \frac{d F_s}{F_s}
\right\rangle + q C_{t,T}^u = y_t +\int_t^T \left\langle \pi_s,
\frac{d F_s}{F_s} \right\rangle + q \int_t^T L(P_s,Z_s^u,u_s) ds
+ q \Phi(P_T,Z_T^u). $$
Using the fact that $P_t = p(t,X_t)$ is a function of the factor process $X$, we
obtain that the value function in (\ref{eq:V}) equals
\begin{equation}\label{ValFunc}
V(t,x,y,z ; q) := \sup_{(u,\pi) \in \mathcal A_t} \mathbb
E_{t,x,y,z} \left[ G \left( X_T, Y_T^{u,\pi},Z_T^{u} \right)
\right],
\end{equation}
where the reward function $G$ is given by
\begin{equation} \label{G}
G(x,y,z) := -\frac{1}{\gamma} e^{-\gamma(y + q \Phi(p(T,x),z))},
\end{equation}
and the state variables $(X,Y^{u,\pi},Z^u)$ evolve as
\begin{equation}\label{dynXYZ}
\left \{
\begin{array}{lcl}
d X_s &=& b(s,X_s) ds +  \Sigma^*(s,X_s) dW_s,\\
dY_s^{u,\pi} &=& \left( \langle \pi_s,\mu_F (s,X_s) \rangle + q
L(p(s,X_s),Z_s^u,u_s) \right) ds + \langle \pi_s, \sigma_F^* (s,X_s) dW_s
\rangle,\\
d Z_s^u &=& u_s ds,
\end{array}
\right.
\end{equation}
with initial conditions $(X_t, Y^{u,\pi}_t,Z^u _t) = (x,y,z)$. Notice that the linear growth properties set in Assumption \ref{ass-coeff} combined with the boundedness of $L$ in Assumption \ref{AssC} give the following estimate for the controlled state process $(X,Y^{u,\pi},Z^u)$:
\begin{equation}\label{estimate} \mathbb E_{t,x,y,z} \left [ \sup_{t \le \tau \le T} \left | \left(X_\tau , Y^{u,\pi}_\tau , Z^u_\tau \right) \right |^p \right] \le C_{u,\pi} (1+|(x,y,z)|^p), \quad t \in [0,T), \quad p\ge 1,  \end{equation}
for some constant $C_{u,\pi}>0$, which depends possibly on the control $(\pi, u)$ and is uniform in $t$.

\begin{remark}
{\rm Observe that the linear growth condition on $b$ and $\Sigma$ (cf. Assumption \ref{ass-coeff}(ii)) imply, through an application of Gronwall's lemma, that
\begin{equation}\label{expX}
\sup_{t\in [0,T]} \mathbb E\left[ e^{\eta | X_t|^2} \right] < \infty,
\end{equation}
for some $\eta >0$.}
\end{remark}

Within this formulation, the UIP of $q \ge 0$ units of the structured product
is the unique solution $v_t = v(t,x,y,z;q)$ (whenever it exists) to
$$
V(t,x,y-v_t,z;q) = V(t,x,y,z;0).
$$

\begin{Rem} \label{reparam}
{\rm In principle, the controls associated to the
virtual storage contract described in Example \ref{ex:virtualSt}
do not satisfy Definition \ref{Def:2.6}, where the control $u_t$
belongs to $[0,\bar u]$ with $\bar u$ constant. However, this
example can be reduced to our setting by simply reparameterizing
the control. In fact, one could define a new control $c$ with
values in $[-1,1]$ such that the old control $u$ satisfies $u_t =
f(c_t,Z_t)$ for a suitable function $f (c,z)$ given by
$$ f(c,z) := \left\{ \begin{array}{ll}
c K_1 \sqrt{\frac{1}{z + Z_b} + K_2},	& 0 \leq c \leq 1, \\
c K_3 \sqrt{z},							& -1 \leq c \leq 0, \\
\end{array} \right. $$
and $Z$ solves
$$ dZ_t = f(c_t,Z_t)\ dt ,\quad Z_0 = z_0.$$}
\end{Rem}
\begin{Rem}\label{PTW-Lud}
\rm{Here, we briefly discuss two papers, \cite{Ludkovski08} and \cite{PorchetTouziWarin}, that do not fit (strictly speaking) the literature on structured products but that are somewhat related. Indeed, they both deal with the pricing of a physical/industrial asset using a UIP approach with an investment component. However, even though the optimization problems in  \cite{Ludkovski08,PorchetTouziWarin} are mathematically similar to the one considered here, the controls affecting the asset are switching controls with finitely many states. Hence their methods, that are based on optimal switching and BSDEs, are different from ours. Finally, our model is more specific than theirs as it is tailor-made for the pricing and hedging of structured contracts on energy.
}\end{Rem}

\section{Characterization of the UIP
with viscosity solutions}\label{sec:UIPvisco}

In this section we characterize, under some further technical assumptions given below, the UIP in terms of viscosity solutions of suitable Cauchy problems. More precisely,  we prove that the log-value
functions for problem (\ref{ValFunc}) with zero initial wealth, defined as
\begin{equation} \label{J}
J(t,x,z;q) := - \frac{1}{\gamma} \log \left( - V(t,x,0, z;q) \right), \quad q \ge 0,
\end{equation}
can be characterized as the unique continuous viscosity
solutions with quadratic growth to a suitable Cauchy problem.
The UIP is obtained from there as the difference between the two log-value
functions, corresponding to the problems with and without the structured products. This is done using some techniques developed in Pham
\cite{Pham02} together with recent results on uniqueness for a
class of second order Bellman-Isaacs equations, established in Da
Lio and Ley \cite{DaLioLey}.\medskip

\subsection{Heuristics on the value function PDE}

In this section we derive, in a heuristic way, the PDE that the
value functions appearing in the definition of UIP are expected to satisfy. It is
a classical property in the presence of an exponential utility function (see, e.g., the
papers \cite{B2,B1,BecSch,HendersonHobson09,VV}) that one can factor out the initial wealth $y$ so that
$$ V(t,x,y,z;q) = e^{- \gamma y} V(t,x,0,z;q), \quad y \in \mathbb R. $$
Hence, by definition of UIP, we deduce
$$ e^{- \gamma (y-v)} V(t,x,0,z;q) = V(t,x,y-v,z;q) =
V(t,x,y,z;0) = e^{- \gamma y} V(t,x,0,z;0) $$
so that the UIP $v$ is given by
\begin{equation}\label{UIP-J}
v = - \frac{1}{\gamma} \log \frac{V(t,x,0,z;q)}{V(t,x,0,z;0)}= J(t,x,z;q) - J(t,x,z;0),
\end{equation}
where $J$ denotes the log-value function defined in (\ref{J}). From the general theory of stochastic optimal control with Markovian state variables it is clear that the value function $V$ is expected to satisfy the following HJB equation
\begin{equation} \label{eq:HJB_V}
V_t(t,x,y,z;q) + \sup_{(u,\pi) \in [0,\bar u]\times \mathbb R^n}
\mathcal L^{u,\pi} V(t,x,y,z;q) = 0,
\end{equation}
with terminal condition $V(T,x,y,z ; q) = G(x,y,z)$ and where  
$$ \mathcal L^{u,\pi} V
= \left( \langle \pi,\mu_F \rangle + q L \right) V_y + \langle b,
V_x \rangle + u V_z + \frac{1}{2} |\pi^{*} \sigma_F^*|^2 V_{yy} +
\frac{1}{2} \mathrm{tr} \left( \Sigma \Sigma^* V_{xx} \right)
+  \pi^{*} \sigma_F^* \Sigma^* V_{xy} $$
is the generator of the state variable $(X,Y,Z)$.
Recalling that $V(t,x,y,z;q) = - e^{- \gamma y - \gamma
J(t,x,z;q)}$, we can easily deduce from (\ref{eq:HJB_V}) the following PDE for the log-value function $J :=J(t,x,y,z;q)$:
\begin{equation}\label{eq:HJB_J_general}
\begin{array}{c}
\displaystyle J_t + \sup_{(u,\pi) \in [0,\bar u]\times \mathbb
R^n} \big[ \langle \pi, \mu_F \rangle + qL + \langle b, J_x
\rangle + u J_z - \frac{1}{2} \gamma |\pi^{*} \sigma_F^*|^2 \\
- \frac{1}{2} \gamma | \Sigma J_x|^2 + \frac{1}{2} \textrm{tr}
\left(\Sigma^* \Sigma J_{xx} \right) - \gamma \pi^{*} \sigma^*_F
\Sigma J_x \big] = 0.
\end{array}
\end{equation}
The Hamiltonian therein is maximised by the control $\hat \pi$, given by
\begin{equation} \label{optimalpi}
\hat \pi^q = (\sigma_F^* \sigma_F)^{-1} \left(\frac{\mu_F}{\gamma}
- \sigma_F^* \Sigma J_x \right).
\end{equation}
Substituting $\hat \pi^q$ into the equation
(\ref{eq:HJB_J_general}) leads to
\begin{equation}\label{eq:HJB_Jpi*multi_general_B}
\begin{array}{c}
\displaystyle J_t + \frac{1}{2 \gamma} \langle (\sigma_F^*
\sigma_F)^{-1} \mu_F,\mu_F\rangle + \langle \bar b , J_x \rangle
+ \sup_{u \in [0,\bar u]} \Big[ u J_z + qL \Big] \\
-\frac{1}{2} \gamma J_x^{*} B J_x + \frac{1}{2} \textrm{tr}
\left(\Sigma^* \Sigma J_{xx}\right) = 0,
\end{array}
\end{equation}
where
\begin{equation}\label{bbar} \bar b := b - \Sigma^* \sigma_F (\sigma_F^* \sigma_F)^{-1}
\mu_F \end{equation}
and $B$ is a $m \times m$ symmetric matrix given by
\begin{equation}\label{eq:B}
B := \Sigma^* \Sigma - (\sigma_F^* \Sigma)^* (\sigma_F^*
\sigma_F)^{-1} (\sigma_F^* \Sigma ) = \Sigma^* (I_d - \sigma_F
(\sigma_F^* \sigma_F)^{-1} \sigma_F^* ) \Sigma.
\end{equation}
The terminal condition for $V$ translates into
\begin{equation} \label{terminalJ}
J(T,x,z;q) = \frac{\log \gamma}{\gamma} + q\Phi(p(T,x),z), \quad
(x,z) \in \mathbb R^m \times [0, \bar u T].
\end{equation}
\begin{Rem} \rm{In order to compute the UIP as in Equation (\ref{UIP}), we first calculate $J(t,x,z;0)$, which
satisfies Equation (\ref{eq:HJB_Jpi*multi_general_B}) with the
terminal condition $J(T,x,z;0) = \frac{\log \gamma}{\gamma}$. It
is a classical and intuitive result that, in this situation,
$J(t,x,z;0)$ does not depend on $z$. Denoting $J(t,x,z;0)$ by
$J^0 (t,x)$ for simplicity, we have that it fulfills
\begin{equation}\label{eq:HJB_J0}
\begin{array}{c}
\displaystyle J_t^0 + \frac{1}{2 \gamma} \langle (\sigma_F^*
\sigma_F)^{-1} \mu_F,\mu_F\rangle + \langle \bar b , J^0_x
\rangle
- \frac{1}{2} \gamma J_x^{0,*} B J_x^0 + \frac{1}{2} \textrm{tr}
\left(\Sigma^* \Sigma J_{xx}^0\right) = 0.
\end{array}
\end{equation}
Thus, subtracting Equation (\ref{eq:HJB_J0}) from Equation
(\ref{eq:HJB_Jpi*multi_general_B}) and using the fact that
$$
- \frac{1}{2} \gamma J_x^{*} B J_x + \frac{1}{2} \gamma J_x^{0,*}
B J_x^0 = - \frac{1}{2} \gamma v_x^* B v_x - \gamma J_x^{0,*} B
v_x
$$
we obtain the following PDE for the UIP $v$:
\begin{equation}\label{eq:HJB_UIPmulti_general}
\begin{array}{c}
\displaystyle v_t + \langle \bar b, v_x \rangle + \sup_{u \in
[0,\bar u]} \Big[ u v_z+ q L \Big] + \frac{1}{2} \mathrm{tr}
\left(\Sigma^* \Sigma v_{xx}\right) 
\displaystyle - \frac{1}{2} \gamma v_x^{*} B v_x - \gamma
J_x^{0,*} B v_x = 0,
\end{array}
\end{equation}
with the terminal condition
\begin{equation}\label{eq:terminalCondHJB_UIP}
v(T,x,z;q) = q \ \Phi(p(T,x),z).
\end{equation}
Notice that solving the HJB equation for the UIP $v(t,x,z;q)$
above requires the knowledge of $J^0$, which is the log-value
function of the optimal investment problem with no claim. This
phenomenon is due to the presence of the non-tradable factors $X$
in the dynamics of the forward contracts $F$ and it has been
observed in a somewhat different model in \cite{B1}, where the
non-tradable factors follow a pure jump dynamics.} 
\end{Rem}

\subsection{Existence and uniqueness results}\label{PDE-UIP}

In this section we show that the
log-value function $J$ is the unique continuous viscosity
solution with quadratic growth of equation
(\ref{eq:HJB_Jpi*multi_general_B}) with the terminal
condition (\ref{terminalJ}). From there, the UIP $v$ is easily
found via the equality (\ref{UIP-J}). We will work under the following standing assumption. Recall that the matrix $B$ has been defined in (\ref{eq:B}).
 
\begin{Ass}\label{MainAss}
The following properties hold: \begin{enumerate}
\item[(i)] $b$ is $C^1$, $B$ and $\langle  \Sigma^* \sigma_F ( \sigma_F ^* \sigma_F)^{-1}
 , \mu_F\rangle$ are $C^1$ and Lipschitz in $x$
uniformly in $t$.
\item[(ii)]  $\mu_F$ is bounded and $\langle (\sigma_F ^* \sigma_F)^{-1} \mu_F ,
\mu_F\rangle$ is $C^1$ and Lipschitz in $x$ uniformly in $t$.
\item[(iii)] $\sigma_F^* \sigma_F$ is bounded and uniformly elliptic, i.e.,
for some $\epsilon > 0$,
\begin{equation} \label{UnifEllF}
(\sigma_F ^* \sigma_F)(t,x) \geq \epsilon I_n, \qquad \mbox{for
all $(t,x) \in [0,T] \times \R^m$.}
\end{equation}
\item[(iv)] The matrix $B$ 
is positive semidefinite and there exists a constant
$\delta >0$ (uniform in $t,x$) such that
\begin{equation} \label{ineqB}
\frac{1}{\delta} | \xi |^2 \le \langle \xi , B \xi \rangle \le
\delta | \xi |^2
\end{equation}
for all vectors $\xi \in \rm{Im}(B)$, the image of $B$.
\end{enumerate}
\end{Ass}

Some comments on these hypotheses are in order. All the assumptions above, with the exception of $C^1$-regularities and boundedness of $\mu_F$ (linear growth is actually sufficient) have to be imposed in order to apply the method and the results established in \cite{DaLioLey}. In particular, condition (iv) on $B$ is related to the coercivity hypothesis in
Assumption A1 in \cite{DaLioLey}, which has a crucial role in the
proof of their comparison theorem. Such a property has to be verified on a case-by-case basis. Some examples where this assumption is verified are provided in Section \ref{ex}.

The additional $C^1$-regularity assumptions as well as the boundedness of $\mu_F$ allow us to adapt results from \cite{Pham02} to get the quadratic growth condition of the log-value function
$J^0$ of the investment problem with no claim. Furthermore, thanks to
Assumption \ref{AssC} on the structured contract, the latter property will be inherited by the log-value function, $J$, with the
claim.\\

We are now ready to state the main result of the paper.

\begin{Thm}\label{MainThm}
Let Assumptions \ref{AssC} and \ref{ass-coeff} hold. Under Assumption \ref{MainAss}, the log-value function $J$,
defined in (\ref{J}), is the unique continuous viscosity
solution with quadratic growth of the Cauchy problem
(\ref{eq:HJB_Jpi*multi_general_B}) with terminal condition
(\ref{terminalJ}).
\end{Thm}

Before proving the theorem, we give a preliminary 
result showing that the value function $V$ is a
(possibly discontinuous) viscosity solution of a
Hamilton-Jacobi-Bellman (HJB) equation in the interior of its
domain. Its proof is postponed to the Appendix.
\begin{Prop}\label{visco}
Let Assumptions \ref{AssC} and \ref{ass-coeff} hold. Under Assumption \ref{MainAss}, the value function $V$ in (\ref{ValFunc}) is a (possibly discontinuous) viscosity
solution of the HJB equation
\begin{equation} \label{eq:HJB_V_general}
V_t(t,x,y,z;q) + \sup_{(u,\pi) \in [0,\bar u]\times \mathbb R^n}
\mathcal L^{u,\pi} V(t, x,y,z ;q) = 0, \qquad (t, x,y,z) \in [0,T)
\times \mathbb R^m \times \mathbb R \times \mathbb R
\end{equation}
with terminal condition $V(T, x,y,z;q) = G(x,y,z )$,
where
$$ \mathcal L^{u,\pi} V
=
\left( \langle \pi,\mu_F \rangle + q L \right) V_y + \langle b,
V_x \rangle + u V_z + \frac{1}{2} |\pi^{*} \sigma_F^*|^2 V_{yy} +
\frac{1}{2} \mathrm{tr} \left( \Sigma \Sigma^* V_{xx} \right)
+  \pi^{*} \sigma_F^* \Sigma^* V_{xy}. $$
\end{Prop}

At this point we are in position to prove Theorem \ref{MainThm}.

\begin{proof}[Proof of Theorem \ref{MainThm}.]
We consider the existence first. This is an easy consequence of
Proposition \ref{visco} above, which gives that the value
function $V$ is a viscosity solution of equation
(\ref{eq:HJB_V_general}). It then suffices to use the definition
of viscosity solution to check that the log-value function $J$
defined in (\ref{J})
is a (possibly discontinuous) viscosity solution of the PDE
(\ref{eq:HJB_Jpi*multi_general_B}).

To complete the proof, it remains to show that $J$ is unique in
the class of all continuous viscosity solutions with quadratic
growth for the Cauchy problem (\ref{eq:HJB_Jpi*multi_general_B}) and (\ref{terminalJ}). The
main idea for uniqueness is to use the comparison theorem in
\cite[Th. 2.1]{DaLioLey}. For reader's convenience, we split
the rest of the proof into two steps.\medskip

(i) \emph{Reduction to Da Lio and Ley \cite{DaLioLey} setting.}
First, we use a Fenchel-Legendre transform to express the
quadratic term in our pricing PDE (\ref{eq:HJB_Jpi*multi_general_B}) into an infimum over the image
of $B$ of a suitable function. More precisely, we apply a
classical result in convex analysis (e.g. \cite[Ch.III, Sect.
12]{Rock}) to get
\begin{equation}
F(w) := - \frac{1}{2} \langle w, B w \rangle = \inf_{\alpha \in
\rm{Im}(B)} \{ - \tilde F (\alpha) - \langle \alpha,w \rangle \} =
\inf_{\alpha \in \mathbb R^m} \{ - \tilde F (\alpha) - \langle \alpha,w
\rangle \}, \label{F}
\end{equation}
for all vectors $w \in \mathbb R^m$, where $\tilde F$ is the conjugate
of $F$, which is also given by 
\[ \tilde F (\alpha) = - \frac{1}{2} \langle
\alpha, B^{-1} \alpha \rangle , \]
when $\alpha \in \rm{Im}(B)$ and
$-\infty$ otherwise. Notice that the first infimum in (\ref{F}) is computed
over the image of $B$ since the matrix $B$ is not necessarily
invertible in our framework.
Using (\ref{F}), we can rewrite equation
(\ref{eq:HJB_Jpi*multi_general_B}) as
\begin{equation}\label{pdeJF}
\begin{array}{c}
\displaystyle J_t + \frac{1}{2 \gamma} \langle (\sigma_F^*
\sigma_F)^{-1} \mu_F,\mu_F\rangle + \langle \bar b, J_x \rangle %
\\
+ \sup_{u \in [0,\bar u]} \Big[ u J_z + qL \Big]
+\gamma F(J_x) + \frac{1}{2} \textrm{tr} \left(\Sigma^* \Sigma
J_{xx}\right) = 0,
\end{array}
\end{equation}
with $\bar b$ as in (\ref{bbar}) and with terminal condition $J(T,x,z;q) = \frac{\log \gamma}{\gamma}
+ q\Phi(p(T,x),z)$.
In order to reduce this PDE to the one in \cite[Eq.
(1.1)]{DaLioLey}, we apply the time reversal
transformation $\widehat J (t,x,z;q) := J(T-t,x,z;q)$, turning the PDE (\ref{pdeJF}) into the following
\begin{equation} \label{eq:HJB_UIP_DLL}
\begin{array}{c}
\displaystyle -\widehat J_t + \frac{1}{2 \gamma} \langle
(\sigma_F^* \sigma_F)^{-1} \mu_F,\mu_F\rangle + \langle \bar b,
\widehat J_x \rangle 
+ \sup_{u \in [0,\bar u]} \Big[ u \widehat J_z + qL \Big]
+\gamma F(\widehat J_x) + \frac{1}{2} \textrm{tr} \left(\Sigma^*
\Sigma \widehat J_{xx}\right) = 0,
\end{array}
\end{equation}
with the initial condition
\begin{equation} \label{ic:DLL}
\widehat J(0,x,z;q) = \frac{\log \gamma}{\gamma} +
q\Phi(p(T,x),z).
\end{equation}
Notice that this Cauchy problem is a particular case of the one
studied in \cite{DaLioLey} since our Assumptions \ref{AssC}, \ref{ass-coeff} and \ref{MainAss} imply Assumptions (A1), (A2), (A3) in
\cite{DaLioLey}. In particular, Assumption
\ref{MainAss}(iv) implies the same property for $B^{-1}$, giving
(A1)(iii) in \cite{DaLioLey}. Indeed on the image of $B$,
$B^{1/2}$ as well its inverse $B^{-1/2}$ are well-defined. Since
$B^{-1/2}: \textrm{Im}(B) \to \textrm{Im}(B)$, we have that,
e.g., the LHS in (\ref{ineqB}) implies $\delta^{-1}\vert B^{-1/2}
y\vert ^2 \leq \langle B^{-1/2} y, BB^{-1/2}y\rangle$ for all
$y\in \textrm{Im}(B)$, leading to $\langle y , B^{-1} y \rangle
\leq \delta \vert y \vert ^2$ for all $y\in \textrm{Im}(B)$. The
other inequality is obtained in a similar way.
\medskip

(ii) \emph{Uniqueness.} We proceed by contradiction. Assume that there exists another continuous viscosity solution
with quadratic growth $\tilde J$ of the Cauchy problem (\ref{eq:HJB_UIP_DLL}) with
terminal condition (\ref{ic:DLL}).
Then, by calling $J^*$ and $\tilde J^*$ their u.s.c. envelopes
and $J_*$ and $\tilde J_*$ their l.s.c. envelopes, we have, by
definition of viscosity solution, that $J^*$, $\tilde J^*$ are
u.s.c. viscosity subsolutions and $J_*$, $\tilde J_*$ are l.s.c.
viscosity supersolutions of equation (\ref{eq:HJB_UIP_DLL}),
obviously with $\tilde J_* \le \tilde J^*$. We also have
$J_*(T,x,z;q) \leq \frac{\log \gamma}{\gamma}+ q \Phi(p(T,x),z)
\leq J^*(T,x,z;q)$, by definition of upper and lower envelopes.
We now want to prove that 
\begin{equation} \label{ineqT} J^*(T,x,z;q) \leq \frac{\log
\gamma}{\gamma}+q \Phi(p(T,x),z) \le J_* (T,x,z;q) , \end{equation}
for all $q \geq 0$, $x \in \mathbb R^m, z\in [0,\bar u T]$. To prove the inequalities (\ref{ineqT}) it suffices to apply Theorem 4.3.2 and subsequent Remark 4.3.5 in \cite{PhamBook09}.\footnote{Notice that in our case the function $G$ appearing in the statement of Theorem 4.3.2 in \cite{PhamBook09} can be chosen to be any positive number.}  

Moreover it can be proved that
$J(t,x,z;q)$ has quadratic growth in $(x,z)$, uniformly in $t$, for all $q\ge 0$ (ref. Lemma
\ref{quadratic_growth} in the Appendix). Then, by the comparison
theorem \cite[Theorem 2.1]{DaLioLey}, we have that
$$ J_* \leq J^* \leq \tilde J_* \leq \tilde J^* \leq J_* $$
on $[0,T] \times \R^m \times \R$. This implies that $J_* = J^* =
J = \tilde J$, and that $J$ is continuous. The proof is therefore
complete.
\end{proof}

As a consequence of the result in Theorem \ref{MainThm}, we have
a good candidate for the optimal hedging strategy,
which is given by
\begin{equation} \label{optimalH}
\hat h^q := \hat \pi^q - \hat \pi^0 = - (\sigma_F^* \sigma_F)^{-1} \sigma_F^* \Sigma v_x,
\end{equation}
where $v_x$ is the gradient with respect to the factor
variables, when it exists, of the UIP. Indeed, the candidate
optimal strategy with or without the structured product in the
portfolio is given by equation (\ref{optimalpi}), where $J =
J(t,x,z;q)$ with $q > 0$ or $q = 0$ in the two cases,
respectively. Thus the optimal hedging strategy $\hat h^q$ is given by
\begin{eqnarray*}
\hat \pi^q - \pi^0 &=& \hat \pi (t,x,z;q) - \hat \pi(t,x,z;0) \\
& =& - (\sigma_F^* \sigma_F)^{-1} \sigma_F^* \Sigma
(J_x(t,x,z;q) - J_x(t,x,z;0))\\
& =&  - ((\sigma_F^* \sigma_F)^{-1}
\sigma_F^* \Sigma v_x )(t,x,z;q) ,
\end{eqnarray*}
in analogy with \cite{B2,B1}. Concerning the optimal exercise policy $\hat u$ of the structured contract, a candidate in feedback form is given by solving the maximization problem
\[ \max_{u \in [0,\bar u]} \left[ uv_z (t,x,z;q) + qL (p,z,u)\right].\]
For an explicit formula, consider the case $L(p,z,u)=u \ell (p,z)$ with $\ell$ bounded. In this case, it is easy to see that
\[ \hat u (t,x,z;q) = \bar u \mathbf 1_{\left[v_z (t,x,z;q) > q \ell (p,z)\right]}. \]
Even though working with viscosity solutions does not allow to justify rigorously the optimality of such controls, we observe that they are consistent with the optimal policies that have been obtained in the past literature for more specific models (see e.g. \cite{BCV13,BLN11}).

\begin{Rem}\label{Lipschitz} \rm{Note that we have worked on the log-value function's PDE
(\ref{eq:HJB_Jpi*multi_general_B}) instead of on the PDE for the
price $v$ (cf. equation (\ref{eq:HJB_UIPmulti_general})). The reason for doing so is that the latter is more delicate to handle due to the fact that it
contains the first derivative $J_x^0$ of the log-value function
with no claim. Applying Da Lio and Ley results directly to
equation (\ref{eq:HJB_UIPmulti_general}) would require a
Lipschitz continuity for $J_x^0$ uniformly in $t$, which is
difficult to have in general. Nonetheless, when this condition is
satisfied as in Cartea-Villaplana (see Subsection \ref{CV}) and in the linear dynamics
model in Example \ref{exCL}, the same arguments go through with less assumptions than in Theorem \ref{MainThm}. Indeed, the boundedness of the payoffs $L$ and $\Phi$ implies that the UIP $v$ is bounded and so it has quadratic growth. Therefore, Lemma \ref{quadratic_growth} is not needed anymore and neither are all the $C^1$-regularities and the boundedness of $\mu_F$ as in Assumption \ref{MainAss}. Under the remaining assumptions and when $\mu_F$ has linear growth in $x$ uniformly in $t$ (replacing its boundedness) we can prove that $v$ is the unique continuous viscosity
solution with quadratic growth to equation
(\ref{eq:HJB_UIPmulti_general}) with terminal condition
(\ref{eq:terminalCondHJB_UIP}). The proof is analogous to that of Theorem
\ref{MainThm}, it is therefore omitted.}
\end{Rem}

\begin{Rem}[Complete market case]  \rm{When the market is complete, i.e. $d = n$ and
$\sigma_F$ has full rank, we have $B=0$ so that Assumption \ref{MainAss}(iv) is trivially satisfied and $J^0_x$ does not appear in the PDE for $v$ anymore. In this case, we can work directly with the PDE for $v$ along the same lines as in the previous Remark \ref{Lipschitz}. Therefore, under Assumptions \ref{AssC}, \ref{ass-coeff} and \ref{MainAss}(i)-(ii)-(iii), one can show that $v$ is the unique
continuous viscosity solution with quadratic growth of the HJB
equation
\begin{equation} \label{HJBcomp}
\begin{array}{c}
\displaystyle v_t + \langle b-\Sigma^* (\sigma_F ^* )^{-1} \mu_F , v_x \rangle + \frac{1}{2}
\mathrm{tr} \left(\Sigma^* \Sigma v_{xx}\right) + \sup_{u \in
[0,\bar u]} \Big[ u v_z+ q L \Big] = 0,
\end{array}
\end{equation}
with terminal condition
\begin{equation} \label{HJBcompterm}
v(T,x,z;q) = q \Phi(p(T,x),z).
\end{equation}
Moreover, one can weaken the boundedness of $\mu_F$ and require only linear growth in $x$ uniformly in $t$. This result extends to our setting previous ones in \cite{BCV13,BLN11,ChenForsyth06,Felix12,TDR09},
which were obtained for particular types of structured contracts, e.g., swings and virtual storages, and without trading in forward contracts.} 
\end{Rem}

\section{Examples}\label{ex}
\subsection{A class of models with two assets and constant correlation}\label{sec:exCL}
In this section we focus on the following incomplete
market model:
\begin{equation}\label{}
\left\{
\begin{array}{rcl}
\displaystyle \frac{dF_t}{F_t} &=& \mu_F (t,
X_t) dt + \bar \sigma_F (t,
X_t) dW^1 _t , \\
\\
dX_t &=& b (t,
X_t) dt + \sigma (t,
X_t) \left( \rho dW^1 _t + \sqrt{1- \rho^2}dW^2_t \right),
\end{array}
\right.
\end{equation}
where $W=(W^1,W^2)$ is a bidimensional Brownian motion and $\rho \in (-1,1)$. This is clearly a particular case of the general model in the previous section with $\sigma^* _F (t,x) = (\bar \sigma_F (t,x), 0)$, $\Sigma^* (t,x) = \sigma(t,x) (\rho , \sqrt{1-\rho^2})$ and $P_t =p(t,X_t)$ for some continuous function $p(t,x)$. This model is a generalization of the usual Black-Scholes model with basis risk (see \cite{Davis06,Henderson02,Monoyios04} among many others), with the additional feature that the non traded asset or factor $X$ can appear in the coefficients of the traded asset $F$. 

We suppose that Assumptions \ref{AssC} and \ref{ass-coeff} are in force. Concerning Assumption \ref{MainAss}, we are going to specialize it to the present setting as follows. Observe first that the quantity $\langle \Sigma^*  \sigma_F (\sigma_F ^* \sigma_F)^{-1} , \mu_F\rangle$ appearing in Assumption \ref{MainAss}(i) reads as
\[ 
\langle \Sigma^*  \sigma_F (\sigma_F ^* \sigma_F)^{-1} , \mu_F\rangle (t,x)= \rho  \mu_F(t,x) \frac{ \sigma(t,x)}{\bar \sigma_F(t,x)}
\]
while the scalar product $\langle (\sigma_F ^* \sigma_F)^{-1} \mu_F ,
\mu_F\rangle$ in Assumption \ref{MainAss}(ii) is
\[ 
\langle (\sigma_F ^* \sigma_F)^{-1} \mu_F ,
\mu_F\rangle (t,x) =  \frac{\mu_F^2(t,x)}{\bar \sigma_F^2(t,x)}.
\]
and $(\sigma_F^* \sigma_F)(t,x)$ in Assumption \ref{MainAss}(iii) corresponds to $(\sigma_F^* \sigma_F)(t,x)= \bar \sigma_F^2(t,x)$. Finally, we have $B(t,x)=(1- \rho^2)\sigma^2(t,x)$.
Hence, Assumption \ref{MainAss} is guaranteed by the conditions listed just below and the general results in Theorem \ref{MainThm} can be safely applied.
\begin{Ass} \label{AssConstCorr}
Let the following properties hold: 
\begin{itemize}
\item[(i)] $b\in C^1$, $\sigma \in C^1$;
\item[(ii)] $\mu_F$ is bounded;
\item[(iii)] $\sigma$ and $\bar \sigma_F$ are bounded and bounded away from zero;
\item[(iv)] $\frac{\mu_F}{\bar \sigma_F} \in C^1$ and it is Lipschitz in $x$ uniformly in $t$.
\end{itemize}
\end{Ass}
In this more specific setting, we can obtain more information on the structure of the value function of the buyer of $q$ units of the structured product provided we have the following
\begin{Ass}\label{AssLip}
Let the log-value function $J^0 _x$ be Lipschitz in $x$ uniformly in $t$.
\end{Ass}
Under this assumption, we do not need to suppose that $\mu_F$ is bounded as in \ref{AssConstCorr}(ii) above. Indeed the considerations in Remark \ref{Lipschitz} apply, so that in particular $\mu_F$ can be a linear function of $x$ as in Example \ref{exCL} below.

Let $C_T^u$ be the payoff of a given structured contract as in (\ref{eq:C_T def}). Inspired by the results in Oberman and Zariphopoulou \cite{OZ03}, which in turn extend El Karoui and Rouge \cite{ElKR} to American
options, we obtain a representation of the UIP of the structured
product $C_T^u$ as the value function of an auxiliary optimization
problem with respect to the control $u$ only, under a suitable
equivalent martingale measure involving the derivative $J_x ^0$
of the log-value function of the problem with no claim, and where
$\gamma$ is replaced by a modified risk aversion $\widetilde
\gamma=\gamma (1-\rho^2)$.\medskip

Let us consider the measure $\mathbb Q ^0$ defined as
 \begin{equation}\label{Q0}
\frac{d\mathbb Q^0}{d\mathbb P} \Big \vert _{\mathcal F_t} : =
D^0 _t : = \exp \left( -\int_0 ^t \theta^* _u dW_u -\frac{1}{2}
\int_0 ^t |\theta_u | ^2 du \right), \quad t \in [0,T],
 \end{equation}
where $W=(W^1 , W^2)^*$ and $\theta$ is given by
\begin{equation}
\theta_t = (\theta_t ^1 , \theta_t ^2)^* = \left(
\frac{\mu_F}{\bar \sigma_F} , \; \gamma \sqrt{1-\rho^2} \sigma J_x
^0 \right)^* (t,X_t).
\end{equation}
Notice that the stochastic exponential is well defined, since
$X$ has continuous paths and $\mu_F / \bar \sigma_F$ is
continuous, so that the stochastic integral $\int_0 ^t \theta_u
^1 dW_u ^1$ is well-defined for every $t$. Moreover, the second
integral $\int_0 ^t \theta^2 _u dW_u ^2$ is also
well-defined thanks to the continuity of $\sigma (t,X_t)$ and
the linear growth of $J_x ^0$ (cf. Lemma
\ref{quadratic_growth}).

Finally, in order for the equation (\ref{Q0}) to define a
probability measure, we need to impose that $\mathbb E[D_T ^0]
=1$. This equality holds true when, for instance, $J_x^0$ is bounded, so that in particular Novikov's criterion applies. More generally, one could use the deterministic criteria proposed in \cite{MU} (e.g. Theorem 2.1 therein).

\begin{Rem}
{\rm In the case when the coefficients of $F$ do not depend on
the state variable $X$, when, e.g. both follows geometric Brownian motions with constant correlation, we have that $J^0 _x \equiv 0$, and
$\mathbb Q^0$ coincides with the minimal entropy martingale
measure. Therefore the measure $\mathbb Q^0$ can be viewed as a
perturbation of the minimal entropy martingale measure (see \cite{frittelli00}) where the
correction involves the log-value function $J^0$ of the optimal
pure investment problem.}
\end{Rem}

In what follows we will need the following preliminary lemma, stating the
dynamics of the spot price under the martingale measure $\mathbb
Q^0$. Its proof is based on a standard application of Girsanov's
theorem, and it is therefore omitted.
\begin{Lem}
Assume $\mathbb E[D_T ^0] =1$. Then the dynamics of $X$ under $\mathbb Q^0$ is given by
\begin{equation}
d X_t = \tilde b (t,X_t) dt + \sigma (t,X_t) d W^0 _t,
\end{equation}
where
\[ \tilde b (t,X_t) := \left( b - \rho \sigma \frac{\mu_F}{\bar \sigma_F} -
\tilde \gamma \sigma ^2 J_x ^0 \right) (t, X_t)\]
and
$$
d W^0_t := \rho dW^1 _t + \sqrt{1- \rho^2}dW^2_t  + \left(\rho \frac{\mu_F}{\bar \sigma_F} + \tilde
\gamma \sigma J_x ^0 \right) (t,X_t) dt
$$
defines a $\mathbb Q^0$-Brownian motion and $\tilde \gamma =
\gamma (1 - \rho^2)$.
\end{Lem}

The following proposition extends to our setting the
characterisation in Oberman and Zariphopoulou \cite[Prop.
10]{OZ03}. 
\begin{Prop}
Let the standing Assumptions \ref{AssC}, \ref{ass-coeff}, \ref{AssConstCorr}(i)-(iii)-(iv) and \ref{AssLip} hold. Then the UIP $v = v(t,x,z;q)$ satisfies
\begin{equation}\label{eq:v_via_MEMM}
v(t,x,z;q) = \sup_{u \in \mathcal U_t} \left( - \frac{1}{\tilde
\gamma} \ln \mathbb E^0 _{t,x,z} \left[ e^{- \tilde \gamma q
C_{t,T}^u} \right] \right),
\end{equation}
where $\mathbb E^0 _{t,x,z}$ denotes the conditional expectation
under $\mathbb Q^0$.
\end{Prop}
\begin{proof}
We prove the result by showing that the candidate function
$$
\tilde v = \tilde v(t,x,z;q) := \sup_{u \in \mathcal U_t} \left(
- \frac{1}{\tilde \gamma} \ln \mathbb E^0 _{t,x,z} \left[ e^{-
\tilde \gamma q C_{t,T}^u} \right] \right)
$$
satisfies equation (\ref{eq:HJB_UIPmulti_general}) with terminal
condition (\ref{eq:terminalCondHJB_UIP}) and we conclude using
the comparison theorem in Da Lio and Ley \cite[Th.
2.1]{DaLioLey}. To this end, write $\tilde v$ as
\begin{eqnarray}
\tilde v(t,x,z;q)
& = & - \frac{1}{\tilde \gamma} \ln ( - w(t,x,z;q)),
\label{def-w}
\end{eqnarray}
with
$$
w(t,x,z;q) : = \sup_{u \in \mathcal U_t} \mathbb E^0_{t,x,z}
\left[ - e^{- \tilde \gamma q C_{t,T}^u }\right].
$$
The value function $w$ above solves the following Cauchy problem
in a viscosity sense
\begin{equation*}
\left\{
\begin{array}{l}
\displaystyle w_t(t,x,z;q) + \sup_{u \in [0,\bar u]} \left[
\mathcal L^{u} w(t,x,z;q) - \tilde \gamma q L(p(t,x),z,u) w(t,x,z;q)
\right] = 0 \\
 w(T,x,z;q) = -\exp (-\tilde \gamma q \Phi(p(T,x),z) )
\end{array}
\right.
\end{equation*}
with
$$
\mathcal L^{u} w = \tilde b w_x + u w_z + \frac{1}{2}
\sigma ^2 w_{xx}.
$$
The corresponding Cauchy problem for $\tilde v$ is immediately
obtained: \begin{equation} \label{cauchy_v}
\left\{
\begin{array}{l}
\displaystyle \tilde v_t(t,x,z;q) + \sup_{u \in [0,\bar u]}
\left[ \tilde{\mathcal L}^{u} \tilde v(t,x,z;q) + q L(p(t,x),z,u)
\right] = 0 \\
\tilde v(T,x,z;q) =  q \Phi(p(T,x),z),
\end{array}
\right.
\end{equation}
with
\begin{eqnarray*}
\tilde {\mathcal L}^{u} \tilde v
& = & \tilde b \tilde v_x + u \tilde v_z + \frac{1}{2}
\sigma^2 \left[ \tilde v_{xx} - \tilde \gamma \tilde v_x^2
\right],
\end{eqnarray*}
which is a particular case of equation
(\ref{eq:HJB_UIPmulti_general}) in this setting. 

To identify $\tilde
v$ with the UIP $v$ and conclude, we need a uniqueness result for the Cauchy problem (\ref{cauchy_v}).
Since $J_x ^0$ is assumed to be Lipschitz in $x$ uniformly in $t$, we can use Remark \ref{Lipschitz} to get the existence of a unique continuous viscosity solution with
quadratic growth to the Cauchy problem (\ref{cauchy_v}). Finally,
the boundedness of the payoff $C^u _{t,T}$ (cf. Assumption \ref{AssC}) clearly implies that
the value function $\tilde v (t,x,z)$ has quadratic growth. Thus
the proof is complete.
\end{proof}

The previous proposition suggests the following approach to
compute the UIP and the corresponding (partial) hedging strategy
of a given structured product in this setting: \begin{itemize}
\item first, solve the pure optimal investment problem $V(t,x,y;0)$ with
no claim;
\item second, compute the $x$-derivative of the log-value function $J^0$
giving the new probability measure $\mathbb Q^0$ as well as the
corresponding dynamics of $X$;
\item finally, solve the maximisation problem in (\ref{eq:v_via_MEMM}),
which is now computed with respect to the control $u$ only; its
value function gives the UIP while its derivative with respect to
$x$ gives the hedging strategy via (\ref{optimalpi}).
\end{itemize}

\begin{Exa}[Linear dynamics model]\label{exCL}
\rm{This example is a slight generalization of the model studied in
\cite[Section 2.2]{CL06}:
\begin{eqnarray}
dF_t & = & F_t\left( (a  - k X_t) dt + \bar \sigma_F dW_t ^1\right),\\
dX_t & = & \delta (\theta - X_t )dt + \sigma \left(\rho dW^1 _t +
\sqrt{1-\rho^2} dW_t ^2\right),
\end{eqnarray}
where $a,k,\bar \sigma_F , \delta, \theta, \sigma$ are real constants,
the correlation $\rho$ belongs to $(-1,1)$, and $(W^1, W^2)$
is a bidimensional Brownian motion as before. Here $F$ represents the
price of a forward contract with maturity $T$ written on a
commodity, whose spot price is $P_t := e^{X_t}$, i.e. $p(t,x)=e^x$ in this case. When $k =
1$ we obtain exactly the model in \cite[Section 2.2]{CL06}.

Notice that if $\bar \sigma_F>0$, $\sigma >0$ and $k=0$, then Assumption \ref{AssConstCorr} holds true, while in the general case when $k \neq 0$ Assumptions \ref{AssConstCorr} (ii) is not satisfied. Nevertheless, as we are going to see, in this example $J_x^0$ is Lipschitz, so that Remark \ref{Lipschitz} applies. Hence we can take $\mu_F$ linear in $x$ as above.

To see that $J_x ^0$ is Lipschitz, consider equation (\ref{eq:HJB_J0}) which in this setting becomes
$$ J_t^0 + \frac{1}{2 \gamma} \frac{(a - k x)^2}{\bar \sigma_F ^2} -
\frac{\rho \sigma}{\bar \sigma_F} (a - k x) J_x^0 + \delta (\theta -
x) J^0_x
- \frac{1}{2} \gamma \sigma ^2 (1 - \rho^2) \left( J_x^0
\right)^2 + \frac{1}{2} \sigma ^2 J_{xx}^0 = 0. $$
Then, in analogy with \cite{BenKar}, one guesses that the
solution $J^0$ has the general form
$$ J^0(t,x) = \alpha(t) + \beta(t) x + \Gamma(t) x^2, $$
such that $J^0(T,x) \equiv \frac{\log \gamma}{\gamma}$. This {\em
ansatz} gives the
system of first order ODEs
$$ \left\{ \begin{array}{l}
\displaystyle \alpha' + \frac{a^2}{2 \gamma \bar \sigma_F ^2} + \left(\delta\theta - \rho
\frac{\sigma}{\bar \sigma_F} a\right) \beta - \frac12
\gamma \sigma  ^2 (1 - \rho^2) \beta^2 + \sigma ^2 \Gamma = 0,
\\
\\
\displaystyle \beta' + \left( \rho k\frac{\sigma}{\bar \sigma_F} -\delta -2\gamma \sigma^2 (1-\rho^2)\Gamma \right)\beta - \frac{ak}{\gamma \bar \sigma_F ^2} + 2 \left (\delta \theta  - \rho
a \frac{\sigma}{\bar \sigma_F} \right)\Gamma  = 0, \\
\\
\displaystyle \Gamma' + \frac{k^2}{2 \gamma \bar \sigma_F ^2} + 2 \left( \rho k
\frac{\sigma}{\bar \sigma_F} -\delta \right) \Gamma - 2 \gamma
\sigma^2 (1 - \rho^2) \Gamma^2 = 0, \\
\end{array} \right. $$
with final condition
$$ \alpha(T) = \frac{\log \gamma}{\gamma}, \qquad \beta(T) = 0,
\qquad \Gamma(T) = 0. $$
The system above is solvable in closed form,
as the third equation is a Riccati equation in $\Gamma$, the
second one is a linear equation in $\beta$, which can be solved
once that $\Gamma$ is known, and, finally, the first one can be
solved in $\alpha$ just by integration. 

Notice that, if the
parameter $k$ appearing in the forward drift is zero then the dynamics of the forward contract does not depend on $X$, so that $J^0$ does not depend on $x$, thus leading to $\beta
\equiv \Gamma \equiv 0$.

Finally, equation (\ref{eq:HJB_UIPmulti_general}) is given in
this case by
$$ \begin{array}{c}
\displaystyle v_t + \left( \delta (\theta - x) - \rho
\frac{\sigma}{\bar \sigma_F} (a - k x) - \gamma \sigma ^2 (1 -
\rho^2) (\beta + 2 \Gamma x) \right) v_x + \frac{1}{2} \sigma
^2 v_{xx} \\
\displaystyle - \frac{1}{2} \gamma \sigma ^2 (1 - \rho^2) v_x^2
+ \sup_{u \in [0,\bar u]} \Big[ u v_z+ q L \Big] = 0,
\end{array} $$
with terminal condition
\begin{equation}\label{eq:terminalCondHJB_UIP_bis}
v(T,x,z;q) = q \ \Phi(e^x,z).
\end{equation}}
\end{Exa}

\subsection{The Cartea-Villaplana model with correlation}\label{CV}
Here we consider a slight generalization of the two factor model for the electricity spot price introduced by Cartea and Villaplana in \cite{CV08}. While the two factors are assumed independent in the original paper \cite{CV08}, here we allow for a possibly non zero (constant) correlation between them.  
We recall briefly the main features of the model. The electricity spot log-price $P_t$ at time $t$ is
decomposed into the sum of two stochastic factors $X^C$ and
$X^D$, i.e.,
$$ P_t = \exp\left( \eta (t) + \alpha_C X^C_t + \alpha_D X^D_t \right),$$
with $\alpha_C < 0$ and $\alpha_D > 0$, where $\eta$ represents a
seasonal continuous deterministic component. The factors $X^i_t$, $i = C,D$,
are Ornstein-Uhlenbeck processes driving, respectively, the
capacity of power plants and the demand of electricity. Their dynamics is given by
$$ d X^i_t = - k^i X^i_t\ dt + \sigma_i(t)\ dW^i_t ,$$
where $k^i$ are constant coefficients, $\sigma_i (t)$ are
deterministic measurable functions of time and each $W^i$, for $i=C,D$, is a
unidimensional Brownian motion such that $ d\langle W^C, W^D
\rangle_t = \rho dt $ with a constant correlation $\rho \in (-1,1)$. Notice that the Cartea-Villaplana
model reduces to the Schwarz-Smith model \cite{SS00} when
$\alpha_C = \alpha_D = 1$ and $k^C = 0$ (or $k^D = 0$).

In this example we work under the following standing assumptions:
\begin{Ass}\label{AssCV}
Let $\sigma_C(t)$ and $\sigma_D (t)$ be continuous, bounded and bounded away from zero over $[0,T]$.
\end{Ass}

Since the interest rate is zero, the price at time $t$ of a forward contract with maturity $T$ can be computed via the usual formula $F_t = \E^\Q[P_T|{\cal F}_t]$, $t\in [0,T]$, for a suitable choice of risk-neutral measure $\mathbb Q$ preserving the Gaussian structure of the model as in \cite[Section 5]{CV08}. Following the approach in \cite{CV08} we can obtain the dynamics of the forward price under the risk-neutral measure $\mathbb Q$ as
\[ \frac{dF_t}{F_t} = \alpha_C e^{-k^C(T-t)}
\sigma_C(t)\ dW_t^{\mathbb Q,C} + \alpha_D e^{-k^D(T-t)} \sigma_D(t)\ dW_t^{\mathbb Q, D}, \]
where $W^{\mathbb Q, C}$ and $W^{\mathbb Q, D}$ are two $\mathbb Q$-Brownian motions with correlation $\rho$. Choosing suitably the market prices of risk as in \cite{CV08} and using Assumption \ref{AssCV}, we can obtain the following forward dynamics under the objective probability $\mathbb P$:
$$ \frac{dF_t}{F_t} = \mu_F (t) dt + \alpha_C e^{-k^C(T-t)}
\sigma_C(t)\ dW_t^C + \alpha_D e^{-k^D(T-t)} \sigma_D(t)\ dW_t^D,
$$
where the drift $\mu_F (t)$ is a bounded function of time.
 
We deal separately with two different situations: the incomplete market case with one forward contract (recall that we have two stochastic factors) and the complete one with two forward contracts.

\subsubsection{The case of one forward
contract}\label{sec:CVoneForward}
In this case the agent is allowed to hedge the structured product by trading only in one forward contract. The Cartea-Villaplana model fits the general setting of Subsection \ref{model} with $X = (X^C,X^D)^*$, whose coefficients are
$$ b(t,x^C,x^D) = \left( \begin{array}{c}
- k^C x^C \\ - k^D x^D \\
\end{array} \right), \qquad
\Sigma^*(t,x^C,x^D) = \left(
\begin{array}{cc}
 \sigma_C(t) & 0 \\
 0 & \sigma_D(t) \\
\end{array}
\right) \cdot \left(
\begin{array}{cc}
 1 & 0 \\
 \rho & \sqrt{1 - \rho^2} \\
\end{array}
\right). $$
Notice that
$\Sigma$ has full rank unless $\rho = \pm 1$, as
$$ \Sigma^* \Sigma = \left( \begin{array}{cc}
  \sigma_C^2 				& \rho \sigma_C \sigma_D\\
  \rho \sigma_C \sigma_D	& \sigma_D^2 \\
  \end{array} \right). $$
Let us consider a forward contract $F$ with maturity $T$. Here $\sigma_F(t,X_t)$
depends only on $t$, so that for simplicity we set $\sigma_F (t) := \sigma_F(t,X_t)$, and we have 
\begin{eqnarray*} 
\sigma_F^*(t) &=& \left(
\begin{array}{cc}
\alpha_C e^{- k^C(T - t)} \sigma_C(t) \quad \alpha_D e^{- k^D(T-t)}
\sigma_D(t) \\
\end{array}
\right) \cdot
\left(
\begin{array}{cc}
 1 & 0 \\
 \rho & \sqrt{1 - \rho^2} \\
\end{array}
\right) \\
&=& \left( \alpha_C e^{- k^C(T - t)} \sigma_C(t) + \rho \alpha_D e^{- k^D(T-t)}
\sigma_D(t) , \sqrt{1-\rho^2}  \alpha_D e^{- k^D(T-t)}
\sigma_D(t) \right).
\end{eqnarray*}
We note that, since the correlation between the spot and forward log-prices is not constant, this model does not fit the setting in Section \ref{sec:exCL}.

In this model the matrix $B$ has rank equal to one. In fact, by definition (cf. equation (\ref{eq:B})) we have
$$ B = \Sigma^* (I_2 - \sigma_F (\sigma_F^* \sigma_F)^{-1}
\sigma_F^*) \Sigma, $$
with
\begin{equation}\label{sigma2F} (\sigma_F^* \sigma_F)(t) =
\alpha_D^2 \sigma^2_D(t) e^{- 2 k^D (T-t)}  + \alpha_C^2 \sigma_C^2(t)
e^{-2k^C(T-t)}  + 2 \rho \alpha_C \alpha_D \sigma_C(t) \sigma_D(t) e^{- (k^C
+ k^D) (T-t)} . \end{equation}
Consider $x = \Sigma^{-1} \sigma_F$. Then $x \neq 0$ and we
have
$$ \langle x, B x\rangle = \sigma_F^* (I_2 - \sigma_F (\sigma_F^*
\sigma_F)^{-1} \sigma_F^*) \sigma_F = \sigma_F^* \sigma_F -
\sigma_F^* \sigma_F (\sigma_F^* \sigma_F)^{-1} \sigma_F^*
\sigma_F = 0. $$
Therefore, working on the image of $B$ in equation
(\ref{F}) is fully justified here, as $\mathrm{rank}(B) = 1$. Now, we show that Assumption \ref{MainAss}(iv) is satisfied in this case. Indeed, a direct computation shows that
\begin{align*}
B&= \kappa (t) \cdot \begin{pmatrix}
\alpha_D^2  e^{-2k^{D}(T-t)} &  -\alpha_C\alpha_D  e^{-(k^C+k^{D})(T-t)}\\
 -\alpha_C\alpha_D e^{-(k^C+k^{D})(T-t)}  & \alpha_C ^2 e^{-2k^{C}(T-t)}
\end{pmatrix}
\end{align*}
where
\[ \kappa (t) := \dfrac{(1-\rho^2)\sigma_C ^2 (t) \sigma_D ^2 (t)}{\alpha_{D}^{2} \sigma_{D}^{2}(t) e^{-2k^{D}(T-t)}+\alpha^{2}_{C} \sigma_{C}^{2}(t) e^{-2k^{C}(T-t)}+2 \rho \alpha_C\alpha_D \sigma_C (t) \sigma_D (t) e^{-(k^{C}+k^{D})(T-t)}}.\]
Hence, the two eigenvalues of $B$ are $\lambda_1 (t) \equiv 0$ and
\[ \lambda_2 (t) = \kappa(t) \left(\alpha_D ^2 e^{-k^D (T-t)} + \alpha_C ^2 e^{-k^C (T-t)} \right) >0 .\] 
By Assumption \ref{AssCV} we have that $\sigma_C(t)$ and $\sigma_D (t)$ are bounded and bounded away from zero over $[0,T]$, yielding $\frac{1}{\delta} \le \lambda_2 (t) \le \delta$ for some $\delta >0$ independent of $t \in [0,T]$. This implies Assumption \ref{MainAss}(iv).

Since in this example the two factors $X^C$ and $X^D$ do not enter in the coefficients of the forward contract dynamics, we expect that the derivative $J^0 _x $ of the log-value function is zero. Indeed, this can be obtained from the PDE (\ref{eq:HJB_J0}) satisfied by $J^0$. Since $\mu_F$ and
$\sigma_F$ do not depend on $X$, such a PDE simplifies to
\begin{equation*}
\displaystyle \displaystyle J^0 _t + \frac{1}{2 \gamma}
\frac{|\mu_F|^2}{|\sigma_F|^2} = 0,
\end{equation*}
which gives
$$ J^0(t) = \frac{\log \gamma}{\gamma} + \int_t^T \frac{1}{2
\gamma} \frac{|\mu_F(s)|^2}{|\sigma_F(s)|^2}\ ds. $$
Therefore $J_x^0 \equiv 0$, and equation
(\ref{eq:HJB_UIPmulti_general}) for the UIP becomes
$$
\displaystyle v_t + \left( b^* - \langle (\sigma_F^*
\sigma_F)^{-1} \sigma_F^* \Sigma ,\mu_F \rangle \right) v_x +
\frac{1}{2} \mathrm{tr} \left(\Sigma^* \Sigma v_{xx}\right) -
\frac{1}{2} \gamma v_x^{*} B v_x + \sup_{u \in [0,\bar u]} \Big[
u v_z+ q L \Big] = 0.
$$ Hence, under Assumption \ref{AssCV}, the considerations in Remark \ref{Lipschitz} apply and give that the UIP $v$ is the unique viscosity solution with quadratic growth of the PDE above.
 
Finally, in this case the candidate optimal hedging strategy is given by $\hat h^q = \hat \pi^q - \hat \pi^0 = - (\sigma_F ^* \sigma_F)^{-1} \sigma_F ^* \Sigma v_x$ as in (\ref{optimalH}), where $\sigma_F ^* \sigma_F$ is as in (\ref{sigma2F}) and  
$$ (\sigma_F^* \Sigma)^*(t) = \left( \begin{array}{c}
\alpha_C e^{-(T-t) k^C} \sigma^2_C(t) + \rho \alpha_D e^{-(T-t)
k^D} \sigma_C(t) \sigma_D(t) \\
\alpha_D e^{-(T-t) k^D} \sigma^2_D(t) + \rho \alpha_C e^{-(T-t)
k^C} \sigma_C(t) \sigma_D(t)
  \end{array}
  \right) .$$

\subsubsection{The case of two forward contracts}

We look now at the much simpler situation where the agent can hedge the structured product by trading in two
forward contracts $F^1$ and $F^2$ with respective maturities
$T_1$ and $T_2$, with $T \leq T_1 < T_2$. Then we have
$$ \sigma_F^*(t) = \left(
\begin{array}{cc}
\alpha_C e^{- k^C(T_1 - t)} \sigma_C(t) & \alpha_D e^{-
k^D(T_1-t)} \sigma_D(t) \\
\alpha_C e^{- k^C(T_2 - t)} \sigma_C(t) & \alpha_D e^{-
k^D(T_2-t)} \sigma_D(t) \\
\end{array}
\right) \cdot
\left(
\begin{array}{cc}
 1 & 0 \\
 \rho & \sqrt{1 - \rho^2} \\
\end{array}
\right).
$$
Of course, in this case $B = 0$, since $\sigma_F$
is invertible. Hence, the market model is complete and we are in the situation described in Remark \ref{HJBcomp}. Analogously to the previous case, it is possible to
find an explicit expression for $J^0$, which is now given by
$$ J^0(t) = \frac{\log \gamma}{\gamma} + \int_t^T \frac{1}{2
\gamma} \langle \mu_F(s), (\sigma^*_F \sigma_F)^{-1}(s), \mu_F(s)
\rangle\ ds. $$
Here again $J_x^0 \equiv 0$, so that Remark \ref{Lipschitz} applies and equation
(\ref{eq:HJB_UIPmulti_general}) for the UIP becomes
$$
\displaystyle v_t + \left( b^* - \langle (\sigma_F^*
\sigma_F)^{-1} \sigma_F^* \Sigma ,\mu_F \rangle \right) v_x +
\frac{1}{2} \mathrm{tr} \left(\Sigma^* \Sigma v_{xx}\right) +
\sup_{u \in [0,\bar u]} \Big[ u v_z+ q L \Big] = 0.
$$
Finally the candidate optimal hedging
strategy is given by $\hat h^q = -(\sigma_F^* \sigma_F)^{-1} \sigma_F^* \Sigma v_x$ as before,
where this time
$$ (\sigma_F^* \sigma_F)^{-1} (\sigma_F^* \Sigma)(t) = \left(
\begin{array}{cc}
\displaystyle \frac{e^{-(T_1-t) k^C}}{\alpha_C
\left(1-e^{(T_1-T_2)(k^C - k^D)}\right) } &
\displaystyle \frac{e^{-(T_1-t) k^D}}{\alpha_D
\left(1-e^{(T_1-T_2)(k^D - k^C)}\right)} \\
  \\
\displaystyle \frac{e^{-(T_2-t) k^C}}{\alpha_C
\left(1-e^{(T_1-T_2)(k^D - k^C)}\right)}
 &
\displaystyle \frac{e^{-(T_2-t) k^D}}{\alpha_D
\left(1-e^{(T_1-T_2)(k^C - k^D)}\right)} \end{array} \right). $$

\section{Numerical results} \label{sec:num} 
In this section we present some numerical applications of our results to swing options (see Example \ref{ex:swing})\footnote{All the numerical tests were performed in
{\sc MATLAB} R2015b.}. We focus on this type of contract for essentially two reasons: first, swing options are the main type of volumetric contracts in commodity markets and, second, we want to compare our results to those in \cite{BLN11}.
 
More specifically, in Subsection \ref{sec:num1} we consider the benchmark case with strike price $K=0$ and minimal cumulated quantity $m=0$ and we compare the prices obtained following the UIP approach to those in \cite{BLN11}; in Subsection \ref{sec:num2} we consider more general swing options with $K > 0$ and $m >0$. 
In both parts, we compute the solution of the relevant PDEs using finite difference schemes, as suggested in \cite{BLN11}.

\subsection{Comparison with the results in Benth et al. \cite{BLN11}}\label{sec:num1}
Here, we compare the UIP, obtained by solving the non-linear PDE (\ref{eq:HJB_UIPmulti_general}), with the risk-neutral price that one can find in the energy market literature (e.g.
\cite{BCV13,BLN11,ChenForsyth06,Felix12,TDR09}). The latter is given in terms of a PDE which is essentially linear, except for the first derivative in $z$ and which has the same form as Equation (\ref{HJBcomp}).  

We consider, as in \cite{BLN11}, one swing option (i.e., we take $q=1$) with parameter values
$$
K = 0, \quad \bar u =1 , \quad T=1, \quad m=0, \quad M =0.5 ,
$$ 
i.e., the control $u$ belongs to $[0,1]$ and the holder faces the problem of picking the most favorable price of the commodity, up to a certain total volume $M$.
We set the risk-free interest rate to zero.
Moreover, in order be as close as possible to the setting considered in Benth et al. \cite{BLN11}, where $Z^u$ is constrained to fulfil $Z^u _T \leq M = 0.5$, we use the penalty function
\begin{equation}\label{penalty}
\Phi (p,z) = \min(0,- C(z - 0.5))
\end{equation}
with $C=1000$. Indeed, the authors in \cite{BCV13} prove that when $C \to \infty$ the
price of a contract with penalty $\Phi$ as in (\ref{penalty}) converges to the price of
a contract with the constraint on $Z^u$ as above.\\
Moreover, with a view towards the comparison with \cite{BLN11}, we choose a special case of the linear dynamics model of Example \ref{exCL} with $k
= 0.01$ and where
\begin{equation}
\label{coeff} 
\delta = 0.4, \quad \sigma = 0.55, \quad \theta
= 3.5, \quad \sigma_F = 0.3, \quad a = 0.03, \quad \rho = 0.5. 
\end{equation}
Finally, the risk-aversion parameter is set to be $\gamma = 0.01$.
\begin{Rem}
\emph{Notice that the coefficients $\delta, \theta$ and $\sigma$ above correspond, respectively, to $\kappa, \mu$ and $\sigma$ in \cite{BLN11}, and they have the same numerical values as in \cite{BLN11}. The remaining coefficients $\sigma_F$ and $a$ refer to the dynamics of the forward contract $F$, which is not part of the model in \cite{BLN11}, and $\rho$ is the correlation between (the
logarithms of) the spot price $P$ and $F$.}
\end{Rem}

We compute both kinds of price (the risk-neutral price and the UIP) for such an option, solving numerically the corresponding PDE via finite difference methodology with a backward time stepping scheme. For the sake of consistency, the approximating domain for the logarithm of the spot price is the same as in \cite{BLN11} and the domain for $Z$ is obviously $[0,\bar u T] = [0,1]$, thus leading to a global domain $\mathcal D:= [0,T] \times [x_{min}, x_{max}] \times [0,1]$, with $x_{min} = \ln(21.6)$ and $x_{max} = \ln (73.9)$.  The mesh and the boundary conditions are the same as in \cite{BLN11}, as well as the numerical approximations of $v_t$ and $v_z$, which are explicit in time, and $v_{xx}$, while there is a difference in the partial derivative with respect to $x$ (see Benth, \cite[Equations (4.9) and (4.10)]{BLN11}): denoting by $v_{i,j}^n$ the approximation of $v(t_n, x_i,z_j;q)$ with $n \in \{ 0, \dots,N \}, i \in \{ 0, \dots, I \}$ and $j \in \{ 0, \dots, J \}$ we have
$$
v_x (t_n, x_i,z_j) \stackrel{\sim}{=} \frac{v_{i+1,j}^{n+1} - v_{i-1,j}^{n+1}}{ 2 \Delta x}, $$
with $\Delta x := \frac{ x_{max} - x_{min}}{I}$. 

We plot in Figure \ref{diff} the prices of the swing contract
at time $t=0.5$, obtained with the two approaches (a similar picture can be provided at any other date). In order to stress the difference between the two prices, we do not plot the surfaces for $z \in [0.25 , 0.5]$ (remember that $M=0.5$). As we can see, the two price
surfaces have similar shapes, even though the ``classical'' procedure slightly
overprices the option with respect to the UIP when the log spot price is high.
The difference between the two prices is clearly due to the risk aversion $\gamma$ and the correlation $\rho$ between the underlying and the forward market, where the buyer can invest. We conclude this part by illustrating in Tables \ref{table1} and \ref{table2} below the effect that those two parameters separately have on the UIP. 
Concerning the dependence of the UIP on $\gamma$, which are summarized in Table \ref{table1}, we choose $x, z$ and $t$ so that the difference between the UIP for $\gamma$ that varies and the UIP for $\gamma = 0.01$ is the biggest possible (on bigger discretization domain with $x_{min} = \ln(0.01)$ and $x_{max} = \ln (500)$), finding $x=6.1903, z= 0.4178, t=0.5$. 
Similarly Table \ref{table2} shows how the UIP varies with $\rho$, for $x=6.0931, z= 0, t=0.5$ and for $\gamma$ fixed to $0.01$.
As we can see, the UIP is decreasing in $\gamma$ while it is increasing in $\rho$. Both effects are very natural, since a higher risk aversion for the buyer is expected to induce a lower price, while one would expect a higher price as the correlation $\rho$ with the forward market increases since this widen the hedging opportunities for the buyer so reducing the risk. The combined effect of $\gamma$ and $\rho$ is not clear in general.

\begin{figure}[h!]
 \centering
\includegraphics[scale=0.8
]{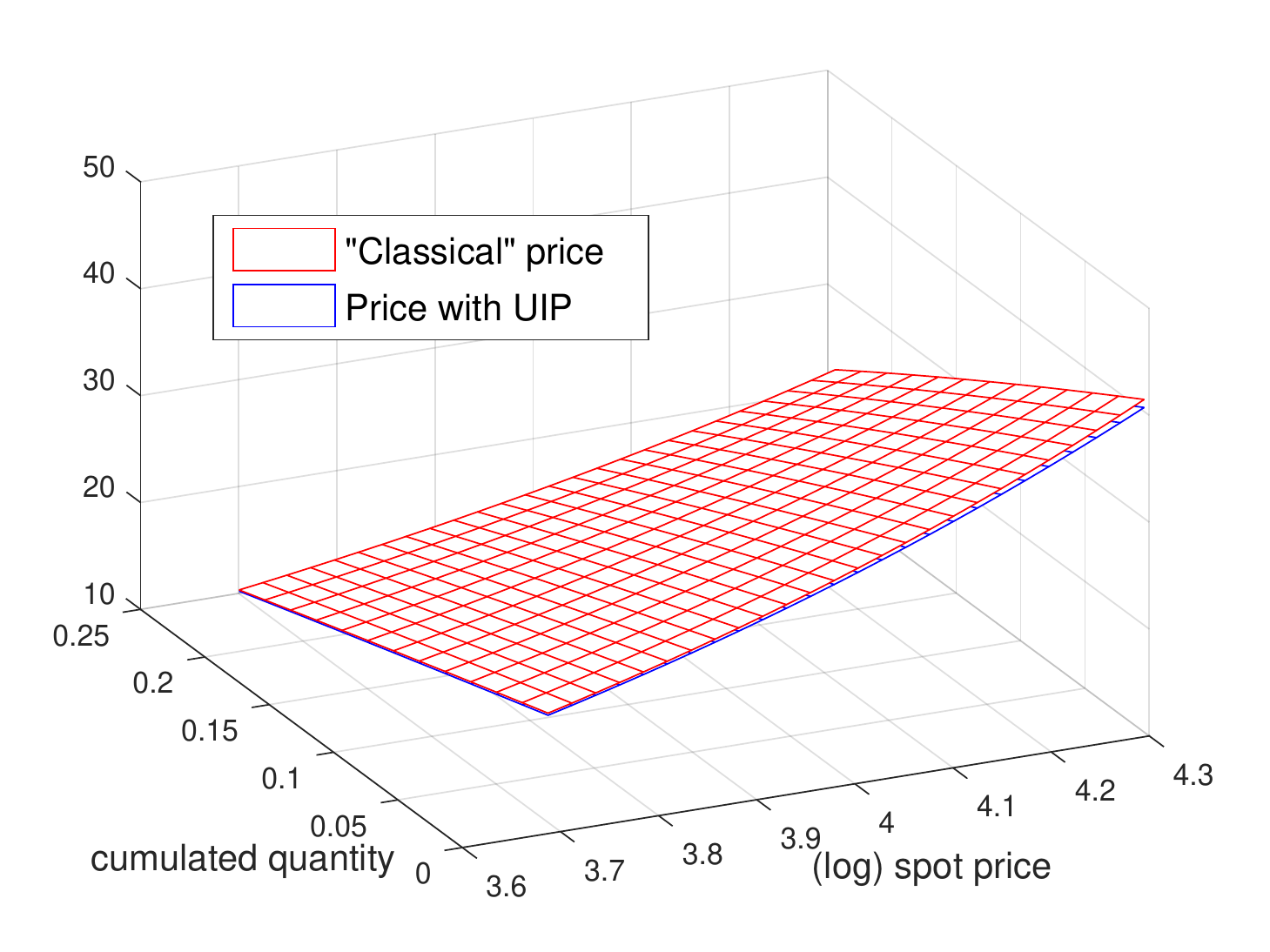}
\caption{``Classical'' price (above) of a swing contract and
UIP (below).  Prices are computed at  $t=0.5$. \label{diff}}
\end{figure}

\begin{table}[h!]
\centering 
\begin{tabular}{|l| c c c c c c |}
\hline
$\gamma$ & 0.01 & 0.02 & 0.04 & 0.06 & 0.08 & 0.10 \\
\hline
\textrm{UIP}  &54.8927 & 52.9527 & 48.3202 & 43.5541 & 39.9692 & 37.7116  \\
\hline
\end{tabular}
\caption{\label{table1} Different values of UIP for a varying $\gamma$ and $x=6.1903, z= 0.4178, t=0.5$.}
\end{table}

\begin{table}[h!]
\centering 
\begin{tabular}{|l| c c c c c |}
\hline
$\rho$ & 0.01 & 0.25 & 0.50 & 0.75 & 0.99 \\
\hline
\textrm{UIP} & 283.6143 & 287.3581 & 300.0573 & 322.1527 & 350.3785  \\
\hline
\end{tabular}
\caption{\label{table2} Different values of UIP for a varying $\rho$ and $x=6.0931, z= 0, t=0.5$.}
\end{table}

\subsection{A more realistic example}\label{sec:num2}
We now focus on computing the UIP of a more realistic swing option contract, with $q=1$,
$$
K = \exp (2.5), \quad \bar u=1, \quad T=1, \quad m=0.1, \quad M=0.5.
$$
Indeed, swing contracts usually have strictly positive strike price and a nonzero minimal cumulated quantity to be purchased. The penalty function we use is the one in equation \eqref{eq:Phi} with $C = 1000$. 
We keep working under the linear dynamics model in Example \ref{exCL}, with $k=0.01$ and with parameters as in \eqref{coeff}. We solve the PDE for $v$ using a backward time stepping finite difference method on the domain $ {\mathcal D} = [0,T] \times [x_{min}, x_{max}] \times [0,1]$, where now we set $x_{min}= \ln(0.01), x_{max}= \ln(500)$. Notice that here $[x_{min}, x_{max}] $ is wider than in the previous subsection, so that the probability that $X$ belongs to this interval is higher. This leads to more stable numerical results.

The approximating schemes for $v_t, v_z,v_x$ and $v_{xx}$ are as in Subsection \ref{sec:num1}, as well as the boundary conditions, 
except for $x=x_{min}$: if $x=x_{min}$ the optimal operational behavior still consists in waiting as long as possible before exercising (this is because $x_{min}$ is much smaller than the expectation of $X$ in the long run and the price is thus expected to increase), but now we have to take into account the constraint $m =0.1$ (recall that $m=0$ in \cite{BLN11}).
Hence we set:
\begin{equation*}
u_s = \left\{
\begin{array}{lcl}
0, && s \in \left( t, T - \frac{(m-z)^+}{\bar u} \right]\\
\bar u, && s \in \left( T - \frac{(m-z)^+}{\bar u}, T \right).
\end{array}
\right.
\end{equation*}
With this choice of $u$, it is possible to explicitly compute the approximating price (recall that in the linear dynamics model in Example \ref{exCL} the spot price is $P_t = e^{X_t}$ and that $\bar u =1$)
$$
\mathbb E_{t,x_{min},z} \left [ \int_t^T u_s (e^{X_s} - K) ds + \Phi(e^{X_T},Z_T^u) \right] = \mathbb E_{t,x_{min},z} \left [ \int_{T - (m-z)^+}^T (e^{X_s} - K) ds + \Phi(e^{X_T},Z_T^u) \right] 
$$ 
as done in Benth et al. \cite[Appendix A]{BLN11}. \medskip

In Figure \ref{prezzi2} we plot the price of the swing option at two different dates. 
\begin{figure}[!h]
 \centering
 \subfigure[Price at $t=0.5$]
  {\includegraphics[width=.48\textwidth
  ]{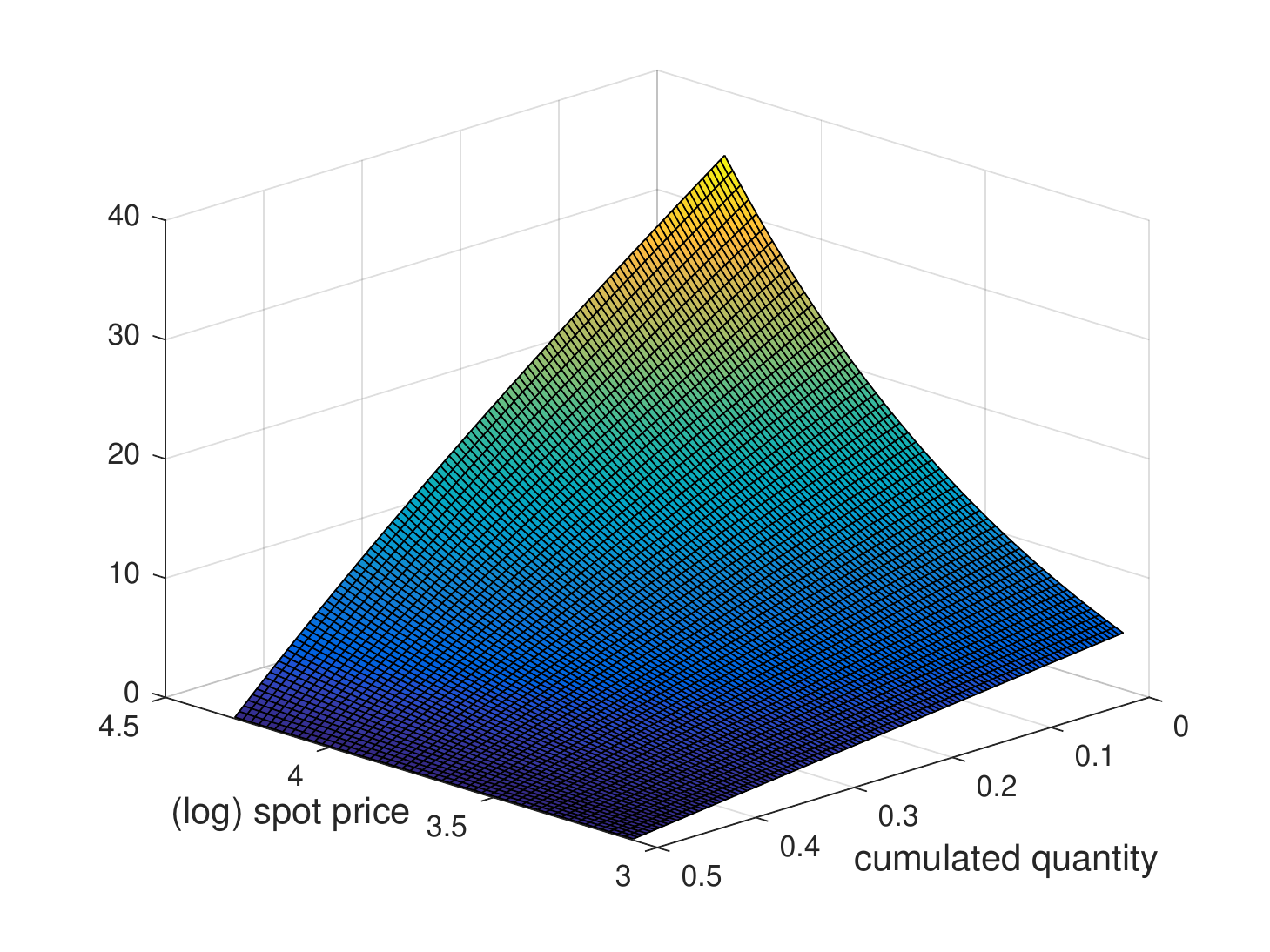} \label{lin}} \hfill
 \subfigure[Price at $t=0.75$]
  {\includegraphics[width=.48\textwidth
  ]{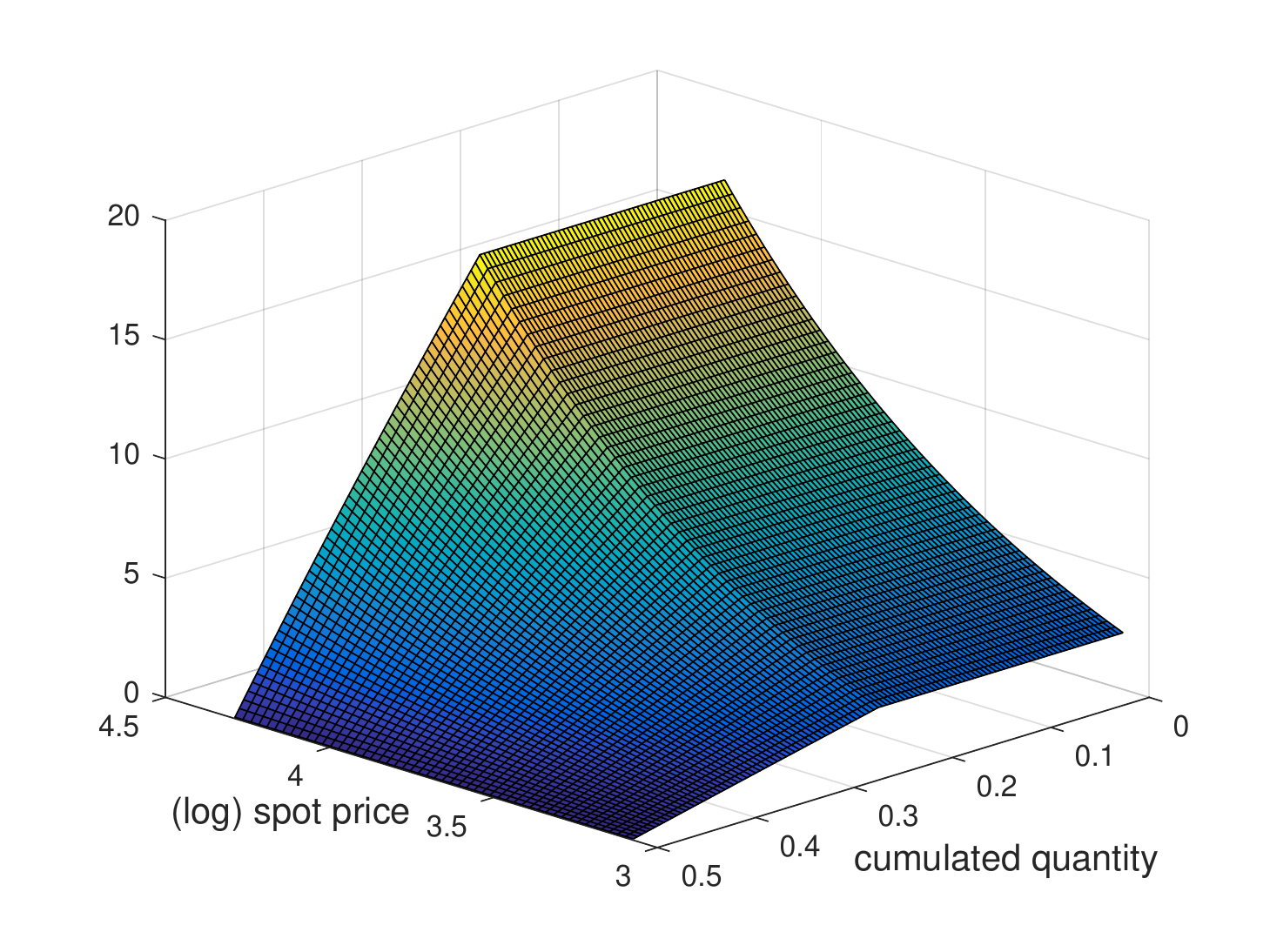} \label{nonlin}} \\
\caption{Price of one swing contract with minimal annual quantity $m = 0.1$ and maximal quantity $M = 0.5$. \label{prezzi2}}
\end{figure}
Notice that in both Figures \ref{prezzi2}(a) and (b) we cut the domain in $z$ in order to focus on positive prices: for $0.5 =M < z < 1$ the penalty function plays a crucial role and the price becomes negative.
We see that the UIP is decreasing in $z$ (as in \cite{BLN11}) and increasing in $x$. Moreover, from Figure \ref{prezzi2}(b) it is clear that for $z > 0.25$ the price is strictly decreasing. This might be explained as follows: for a fixed value of the log spot $x$ and for $t=0.75$, if $z>0.25$ the value of the contract is lower than when $z \le 0.25$ and it even becomes lower and lower as $z$ increases, since the time to maturity is equal to $0.25$ and so if $z>0.25$ the buyer has less opportunities to exercise the option, hence less possibilities to take advantage of (possibly) higher prices.

Moreover, as an example, in Figure \ref{ControlloOttimo} we show the optimal exercise strategy $\hat u$ at time $t=0.75$ as a function of the (log) spot price $x$ and of the cumulated quantity $z$. In the grey region $\hat u = \bar u$, while in the white region $\hat u =0$. 
\begin{figure}[!h]
 \centering
\includegraphics[scale=0.35]{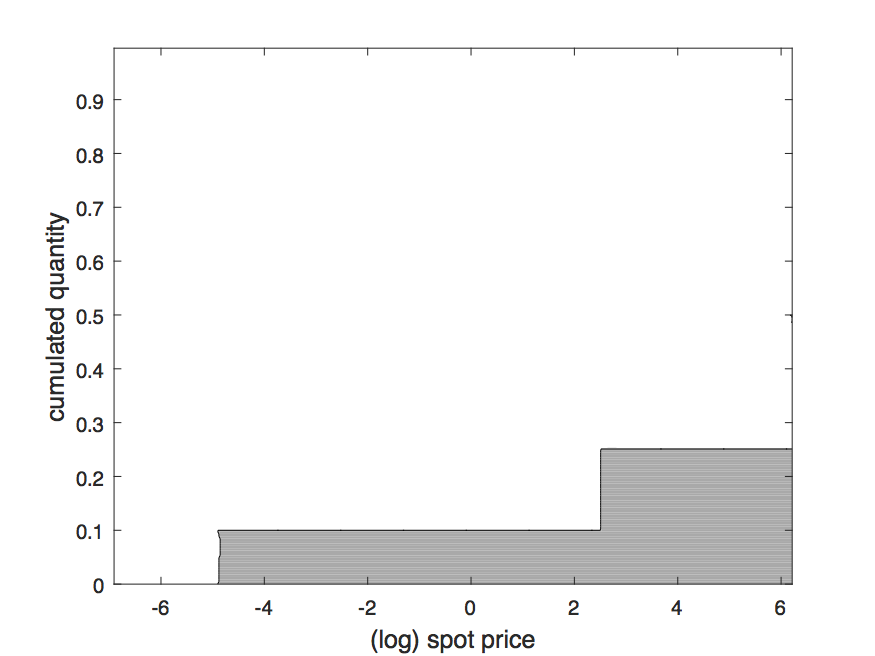}
\caption{Optimal control $\hat u$ at time $t=0.75$. In the grey region $\hat u = \bar u=1$, in the white region $\hat u =0$. \label{ControlloOttimo}}
\end{figure}
From Figure \ref{ControlloOttimo} it is clear that, unless the spot price is very low, if the cumulated quantity $z < m=0.1$, then it is always optimal to exercise the option, to avoid the penalty. Furthermore, when $ x > 2.5$, equivalently the spot price $e^x$ is bigger than the strike price $ K= \exp(2.5)$ and so the optimal policy consists in exercising the option (i.e., $\bar u = 1$) whenever $z \in [0,0.25]$. On the other hand, if the spot price is higher than the strike, $x > 2.5$, and if the cumulated quantity satisfies $z>0.25$ then it is not optimal to exercise the option: in the current state $m < z < M$, thus we are not incurring the penalty and the more we have used of our control, the higher the spot price has to be before we are willing to exercise.

We conclude this section by showing in Figure \ref{PiOttimo} the candidate optimal hedging strategy $\hat h^1$ found in equation \eqref{optimalH} as a function of the (log) spot price $x$ and of the cumulated quantity $z$, at time $t=0.5$.
\begin{figure}[!h]
 \centering
\includegraphics[scale=0.6]{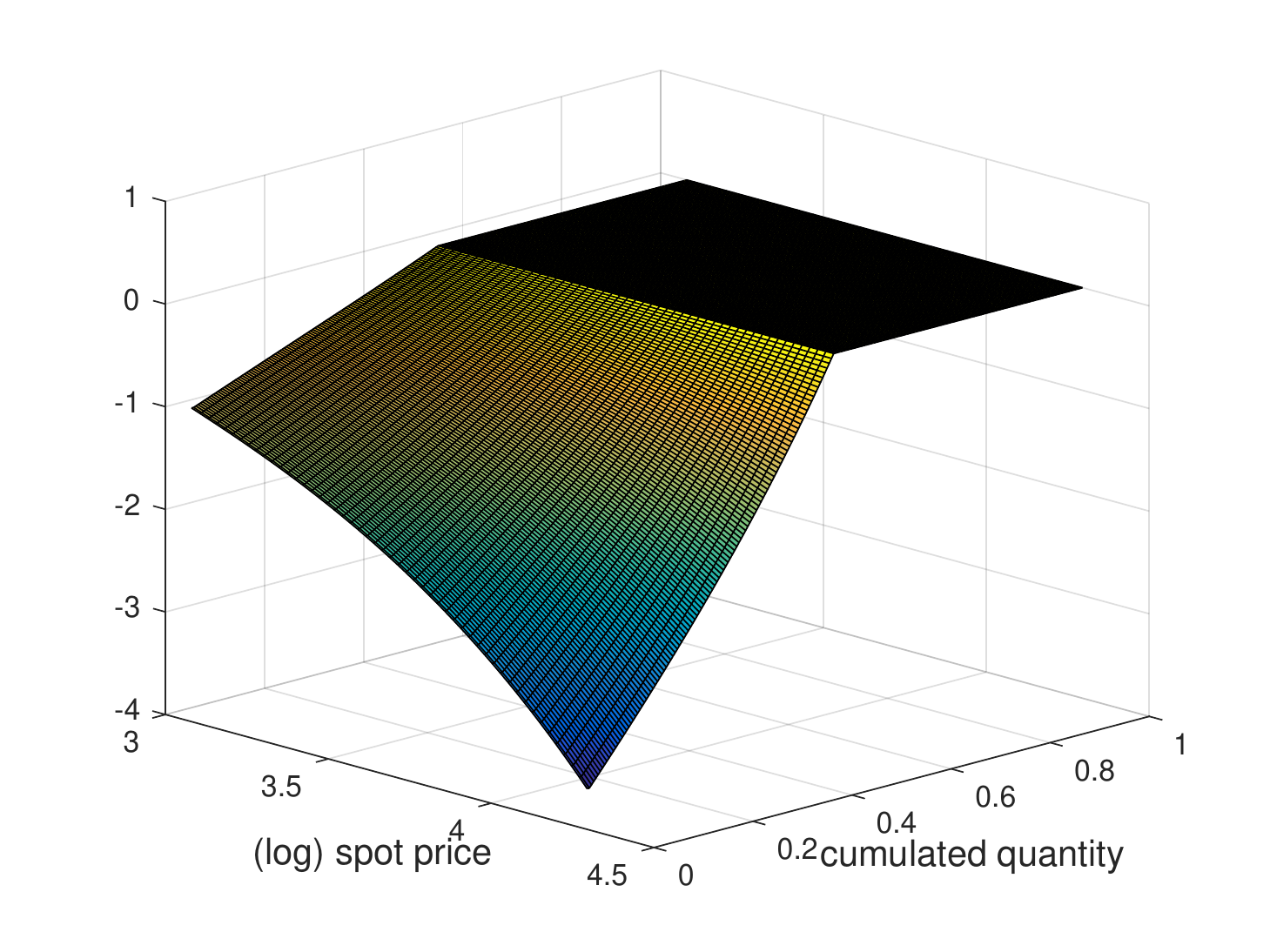}
\caption{Candidate optimal hedging strategy $\hat h^1$ at time $t=0.5$. \label{PiOttimo}}
\end{figure}
We notice that, being the UIP increasing in $x$, $v_x$ is positive on our domain (recall equation \eqref{optimalH}), so that $\hat h^1$ is always negative: in order to hedge a buyer position in a swing option it is always ``optimal'' to sell the forward contract. Moreover, for a fixed $z$, as the (log) spot price increases, the quantity of forward contracts to sell increases. On the other hand, for a fixed $x$, $\hat \pi$ is increasing as a function of $z$, for $z \in [0,0.5]$ (meaning that as the cumulated quantity $z$ increases towards $M=0.5$, selling forward contracts is less and less needed), while $\hat h^1 =0$ for $z \ge M=0.5$, as expected. 

\section{Conclusions}\label{conclusion}
In this paper, we considered the problem of pricing and hedging of structured products in
energy markets from a buyer's perspective using the (exponential) utility
indifference pricing approach. The main novelty with respect to the existing literature is that buyer has the possibility to trade in the forward market in order to hedge the risk coming from the structured contract.

We characterized the UIP in terms of continuous
viscosity solutions of a suitable nonlinear PDE. As a
consequence, we were able to identify a candidate for the optimal exercise strategy of
the structured product as well as a portfolio strategy 
partially hedging the financial position.

Moreover, in a more specific setting with two assets and
constant correlation, we showed that the UIP equals the value function
of an auxiliary simpler optimization problem under a risk neutral
probability, that can be interpreted as a
perturbation of the minimal entropy martingale measure. 

Finally, we provided some numerical applications in the case of swing options. In particular, we computed the UIP price as well as the optimal exercise and hedging strategies for a buyer of one swing option in the linear dynamics model, by solving the corresponding nonlinear PDEs via finite difference schemes. We highlighted the differences with respect to the classical price as in \cite{BLN11} and discussed some qualitative properties.

\appendix

\section{Proof of Proposition \ref{visco}}
The maximisation problem (\ref{ValFunc}) fits the setting of Section 5 in the paper \cite{BT11} on weak dynamic programming principle. In particular, their Corollary 5.6 applies. More precisely, the essential ingredients in the proof of Corollary 5.6 are the a-priori estimate (5.2) in \cite{BT11}, the local boundedness of the value function and the lower semi-continuity of the objective function in $(t,x,y,z)$ for all admissible controls. First, the a-priori estimate holds due to (\ref{estimate}). Concerning the local boundedness of the value function, it can be easily checked that in our setting the
value function is bounded since it is trivially nonpositive and,
being $(u,\pi )=(0,0)$ an admissible strategy, we have
\[ V(t,x,y,z;q) \ge -\frac{1}{\gamma} \exp\left\{-\gamma \left[ y+ q
\inf_{p \in \mathbb R} \left( (T-t) L(p,0,0) + \Phi(p,0) \right)
\right]\right\} > -\infty\]
since the functions $L$ and $\Phi$ are bounded (cf. Assumption \ref{AssC}). Let $(u,\pi)$ be an admissible given control. Since the control is now fixed, we drop it from the notation of the state variable at maturity and denote them as $A_T ^{t,a} := (X_T ^{t,x},Y_T ^{t,a}, Z_T ^{t,a})$ with $a =(x,y,z)$, to stress the dependence on the initial data. Now consider the objective function
\[[0,T] \times \mathbb R^m \times \mathbb R \times [0,\bar uT] \ni (t,x,y,z)=(t,a) \mapsto \mathbb E [ G(A_T ^{t,a})],\]
where $G$ is defined in (\ref{G}).
From the continuity of the function $G$ and of the state variables $A^{t,a} _T$ with respect to the initial data $(t,a)$, we get that $G(A^{t,a} _T )$ is also continuous in $(t,a)$. Moreover, notice that since $L$ and $\Phi$ are bounded (ref. Assumption \ref{AssC}) we have
\[ | G(A_T ^{t,a}) | \le C \exp \left(-\gamma \left(y + \int_t ^T \left\langle \pi_s , \frac{dF_s}{F_s}\right\rangle\right)\right),\]
for some constant $C>0$. Therefore, to prove the lower semi-continuity of the objective function it suffices to show that the family of random variables
\[ \left\{  \exp \left(-\gamma \int_t ^T \left\langle \pi_s , \frac{dF_s}{F_s}\right\rangle \right) : t \in [0,T] \right\}\]
is uniformly integrable. We prove that they are bounded in $L^2$ for all admissible controls, i.e.
\[ \sup_{t\in [0,T]} \mathbb E \left [  \exp \left(-2\gamma \int_t ^T \left\langle \pi_s , \frac{dF_s}{F_s}\right\rangle \right) \right ] < \infty ,\]
which will imply the uniform integrability. Let $\mathcal F_{t,T}$ be the smallest filtration generated by the Brownian increment after $t$ and satisfying the usual conditions. Consider the following change of measure on $\mathcal F_{t,T}$:
\begin{equation}\label{dQdP} \frac{d\mathbb Q _t}{d\mathbb P} := \exp \left( -2\gamma \int_t ^T \pi^*_s \sigma^*_F (s,X^{t,x}_s) dW_s - 2\gamma^2 \int_t ^T | \pi^*_s \sigma^*_F(s,X^{t,x}_s) | ^2 ds \right),\end{equation}
which is well defined since it satisfies the criterion in \cite{LS1}, pp. 220-221, thanks to the boundedness of $\sigma_F^* \sigma_F$ (cf. Assumption \ref{MainAss} (iii)) and the admissibility property (\ref{adm}). Moreover, the change of measure (\ref{dQdP}) satisfies $\sup_{t\in [0,T]} \mathbb E [ (d\mathbb Q_t / d\mathbb P)^2 ] < \infty$.  This is a consequence of the admissibility of $\pi$ as in (\ref{adm}). Indeed,
\[
\frac{d\mathbb Q_t}{d\mathbb P}  \le  \exp \left(-2\gamma \int_t ^T  \pi^*_s \sigma^*_F (s,X^{t,x}_s) dW_s \right),
\]
giving that
\begin{eqnarray*}
\mathbb E\left[ \left(\frac{d\mathbb Q_t}{d\mathbb P}\right)^2 \right] & \le & \mathbb E \left[ \exp \left(-4\gamma \int_t ^T  \pi^*_s \sigma^*_F (s,X^{t,x}_s) dW_s \right)\right] \\
& \le & \mathbb E\left[\exp \left(-8\gamma \int_t ^T  \pi^*_s \sigma^*_F (s,X^{t,x}_s) dW_s - 2\delta \int_t ^T | \pi^*_s \sigma^*_F(s,X^{t,x}_s) | ^2 ds \right)\right]^{1/2} \\
&& \times \mathbb E\left[ e^{2\delta \int_t ^T | \pi^*_s \sigma^*_F(s,X^{t,x}_s) | ^2 ds}\right]^{1/2} \\
& = & \mathbb E\left[ e^{2\delta \int_0 ^T | \pi^*_s \sigma^*_F(s,X^{t,x}_s) | ^2 ds}\right]^{1/2},
\end{eqnarray*}
with $\delta$ such that $2\delta = (8\gamma)^2 /2$, since the first exponential in the second inequality above is a true martingale. Moreover, since $\pi$ is admissible we have $\mathbb E\left[ e^{2\delta \int_0 ^T | \pi^*_s \sigma^*_F(s,X^{t,x}_s) | ^2 ds}\right] < \infty$. 
As a consequence, we obtain that $d\mathbb Q_t / d\mathbb P$ is square integrable. 

Therefore we have
\begin{eqnarray*} && \mathbb E \left [  \exp \left(-2\gamma \int_t ^T \left\langle \pi_s , \frac{dF_s}{F_s}\right\rangle \right) \right ] \\
&& \quad = \mathbb E_{\mathbb Q_t} \left[  \exp \left( -2\gamma \int_t ^T \pi^* _s \mu_F(s,X^{t,x}_s) ds + 2\gamma^2 \int_t ^T | \pi_s ^* \sigma^* _F (s,X^{t,x}_s) |^2 ds \right) \right] \\
&& \quad \le \mathbb E \left [ \left(\frac{d\mathbb Q_t}{d\mathbb P} \right)^2 \right] \mathbb E \left[  \exp \left( -4\gamma \int_t ^T \pi^* _s \mu_F(s,X^{t,x}_s) ds + 4\gamma^2 \int_t ^T | \pi_s ^* \sigma^* _F (s,X^{t,x}_s) |^2 ds \right) \right].
\end{eqnarray*}
Using the linear growth condition of $\mu_F$ and the boundedness of $\sigma_F ^* \sigma_F$ (cf. Assumption \ref{MainAss}(ii) and (iii)), we have
\[ \exp \left( -4\gamma \int_t ^T \pi^* _s \mu_F(s,X^{t,x}_s) ds + 4\gamma^2 \int_t ^T | \pi_s ^* \sigma^* _F (s,X^{t,x}_s) |^2 ds \right) \le \exp \left( \int_0 ^T \left(c_1 | \pi_s | + c_2 |\pi_s|^2 + c_3 |X_s|^2 \right) ds \right),\]
for some positive constants $c_1,c_2,c_3$. To conclude it suffices to prove that the RHS above is integrable for $\mathbb P$. This follows from the admissibility of $\pi$ as in (\ref{adm}) and the exponential uniform bound (\ref{expX}) for $X$.

Finally, even though the space of admissible controls in our setting is smaller than the one in \cite{BT11}, the value functions are the same since any controls in their space $\mathcal U_0$ can be clearly approximated by admissible controls in $\mathcal A$ through truncation. The result follows.

\section{Regularity properties of the log-value function}

In order to prove the next lemma we follow closely the
approach in Pham \cite{Pham02}, which has also been used in \cite{mnif07} in a slightly different model with stochastic volatility with jumps and for an agent with exponential utility. Since the proof mimicks closely the arguments in \cite{Pham02}, we only sketch them pointing out the main differences.

\begin{Lem}\label{quadratic_growth}
Let $q\geq 0$. Let Assumptions \ref{AssC} and \ref{ass-coeff} hold. Under Assumption \ref{MainAss} the log-value function $J(t,x,z;q)$ defined as in
(\ref{J}) has quadratic growth in $(x,z)$ uniformly in $t$.
\end{Lem}
\begin{proof}
Since the claim $C^u _{t,T}$ is bounded in $(x,z)$ uniformly in
the controls $u$ (cf. Assumption \ref{AssC}), it suffices to prove that $J^0 (t,x)$, the
log-value function of the pure investment problem, has quadratic
growth in $x$ uniformly in $t$.

First of all, repeating exactly the same arguments as in the
proof of Theorem 3.1 in \cite{Pham02}, we get that if the PDE
(\ref{eq:HJB_J0}) with terminal condition $J^0(T,x) = \frac{\log
\gamma}{\gamma}$ admits a unique solution belonging to
$C^{1,2}([0,T) \times \mathbb R^m)\cap C^0 ([0,T]\times \mathbb
R^m)$, whose $x$-derivative has linear growth, then such a
solution coincides with $J^0 (t,x)$.

To conclude the proof, we need to show that the PDE
(\ref{eq:HJB_J0}) has a unique smooth solution as above, whose
$x$-derivative has linear growth. We adapt to our setting
the arguments in the proof of \cite[Th. 4.1]{Pham02} under his
Assumptions (H3a). Indeed, notice that our Assumption \ref{MainAss}(i), together with the Lipschitz continuity of $b$ postulated in Assumption \ref{ass-coeff}(ii), corresponds to $\textbf{(H3a)}(i)$ in \cite{Pham02}. Moreover Assumption \ref{MainAss}(ii) implies $\textbf{(H3a)}(ii)$, while Assumption \ref{MainAss}(iii) guarantees $\textbf{(H2)}(b)$ (see Remark 2.3 in \cite{Pham02}).

Consider the PDE (\ref{pdeJF}) in the case $q=0$, with
$F(w)$ replaced by
\begin{equation}
F_k (w) := \inf_{\alpha \in \mathcal B_k} \left\{-\tilde F (\alpha) -
\langle \alpha, w \rangle \right\} , \quad w\in \mathbb R^m
,\label{Fk}
\end{equation}
where $\mathcal B_k$ is the centered ball in $\mathbb R^m$ with
radius $k\geq 1$. Recall that $\tilde F$ is the convex conjugate of $F$ and that is given by
\[ \tilde F(\alpha) = -\frac{1}{2} \langle \alpha , B^{-1} \alpha \rangle, \quad \alpha \in \textrm{Im}(B),\]
while it equals $-\infty$ otherwise. Proceeding as in the proof of \cite[Th.
4.1]{Pham02}, we can apply Theorem 6.2 in \cite{fleming_rishel},
giving the existence of a unique solution $J^{0,k} \in
C^{1,2}([0,T) \times \mathbb R^m)\cap C^0 ([0,T]\times \mathbb
R^m)$ with polynomial growth in $x$, for the parabolic PDE
\begin{equation}\label{pdeJFk}
\begin{array}{c}
\displaystyle J^{0,k}_t + \frac{1}{2 \gamma} \langle (\sigma_F^*
\sigma_F)^{-1} \mu_F,\mu_F\rangle 
+\gamma F_k (J^{0,k}_x) + \frac{1}{2} \textrm{tr} \left(\Sigma^*
\Sigma J^{0,k}_{xx}\right) = 0,
\end{array}
\end{equation}
with terminal condition $J^{0,k} (T,x) = \frac{\log
\gamma}{\gamma}$. Notice that the convex conjugate $\tilde F$ of $F$, appearing in the definition of $F_k
(w)$ in (\ref{Fk}), can take the value $-\infty$, which
is not a problem here since this value does not contribute to the
infimum over $\alpha$.

The next step consists, as in \cite{Pham02}, in using a
stochastic control representation of the solution $J^{0,k}$ to
derive a uniform bound on the derivative, independently of the
approximation. Indeed, from standard verification arguments we
get that
\[ J^{0,k}(t,x) = \inf_{\alpha \in \mathbb B_k} \mathbb
E^{\mathcal Q} \left[\int_t ^T \Lambda (s, X_s , \alpha_s )
ds \mid X_t =x\right],\]
where
\[ \Lambda (s,x,\alpha) = \frac{1}{2 \gamma} \langle
(\sigma_F^* \sigma_F)^{-1} \mu_F, \mu_F\rangle (s,x) - \gamma
\tilde F (\alpha), \]
where $\mathbb B_k$ is the set of $\mathbb R^m$-valued adapted
processes $\alpha$ bounded by $k$, and the controlled dynamics of
$X$ under $\mathcal Q$ is given by
\[ dX_s = (\bar b(s,X_s) - \gamma \alpha_s ) ds +
\Sigma^* (s,X_s) dW^{\mathcal Q}_s ,\]
where $W^{\mathcal Q}$ is a $d$-dimensional Brownian motion under
$\mathcal Q$ and $\bar b$ has been defined in (\ref{bbar}). Notice that, since $\Lambda$ takes the value
$-\infty$ outside the image of $B$, then the optimal Markov
control evaluated along the optimal path $\hat \alpha (s,\hat
X_s)$ will lie on $\mathrm{Im} (B)$ a.s. for every $s\in [t,T]$.
We can use Lemma 11.4 in \cite{FleSon} and the same estimates as
in \cite[Lemma 4.1]{Pham02} to obtain
\[ | J_x ^{0,k} (t,x) | \le C (1+ |x|) , \quad \forall (t,x) \in
[0,T]\times \mathbb R^m , \]
for some positive constant $C$, which does not depend on $k$. Now
we argue as in the proof of \cite[Th. 4.1]{Pham02}, Case (H3a),
to deduce that $|\hat \alpha_k (t,x)| \le C$ for all $t\in [0,T]$
and $|x| \le M$ for some positive constant $C$ (independent of
$k$) and an arbitrarily large $M>0$. Therefore, we get that, for
$k \le C$, $F_k (J^{0,k}_x) = F(J^{0,k}_x)$ for all $(t,x) \in
[0,T]\times \mathcal B_M$. Letting $M$ tend to $+\infty$, we
finally get that $J^{0,k}$ is a smooth solution with linear
growth on derivative to the PDE (\ref{pdeJF}) (with $q=0$). To
conclude, we have that $J^0 = J^{0,k}$ for $k$ sufficiently
large, giving, in particular, that $J^0$ has quadratic growth in
$x$ uniformly in $t$. Therefore the proof is complete.
\end{proof}


\begin{thebibliography}{99}

\bibitem{Aidbook} R. A\"id. Electricity Derivatives. SpringerBriefs in Quantitative Finance. Springer, 2015. 

\bibitem{ACLP} R. A\"id, L. Campi, N. Langren\'e, H. Pham.
\textit{A Probabilistic Numerical Method for Optimal Multiple
Switching Problems in High Dimension}. SIAM Journal on Financial
Mathematics 5 (1) (2014), 191--231.

\bibitem{BCV13} M. Basei, A. Cesaroni, T. Vargiolu. \textit{Optimal exercise of swing contracts in energy markets: an integral constrained stochastic optimal control problem}. SIAM Journal on Financial Mathematics 5(1) (2014), 581--608.

\bibitem{B2} D.~Becherer. \textit{Rational hedging and valuation
of integrated risks under constant absolute risk aversion}.
Insurance: Math. and Econ. 33 (2003), 1--28.

\bibitem{B1} D.~Becherer. \textit{Utility-Indifference Hedging
and Valuation via Reaction-Diffusion Systems}. Proc. Royal
Society A, 460 (2004), 27--51.

\bibitem{BecSch} D.~Becherer, M.~Schweizer. \textit{ Classical
solutions to reaction-diffusion systems for hedging problems with
interacting It\^o and point processes}. The Annals of Applied
Probability 15 (2) (2005), 1111--1144.

\bibitem{BC} G. Benedetti, L. Campi. \textit{Utility indifference
valuation for non-smooth payoffs with an application to power
derivatives}, 2015. Applied Mathematics and Optimization, forthcoming.

\bibitem{BE} F.~E.~Benth, M.~K.~V.~Eriksson. {\em Energy derivatives with volume control}. Chapter 16 of {\em Handbook of Risk Management in Energy Production and Trading}, ed. R.~Kovacevic {\em et al.}. International Series in Operations Research \& Management Science, Vol.~199, 413--432. Springer 2013.

\bibitem{BenKar} F.~E.~Benth, K.~H.~Karlsen. \textit{A note on
Merton's portfolio selection problem for the Schwartz
mean-reversion model}. Stochastic Analysis and Applications 23
(4) (2005), 687--704.

\bibitem{BT11} B. Bouchard, N. Touzi. \textit{Weak dynamic programming principle for viscosity solutions}. SIAM Journal on Control and Optimization 49(3), (2011), 948--962.

\bibitem{BLN11} F.~E.~Benth, J.~Lempa, T.~K.~Nilssen. \textit{On
the optimal exercise of swing options in electricity markets}.
The Journal of Energy Markets 4 (4) (2012), 1--27.

\bibitem{CL06} R.~Carmona, M.~Ludkovski. \textit{Pricing
commodity derivatives with basis risk and partial observation},
2006, preprint. Available online at {\tt
http://www.pstat.ucsb.edu/faculty/ludkovski/CarmonaLudkovskiBasis.pdf}

\bibitem{CL10} R. Carmona, M. Ludkovski. \textit{Valuation of energy storage: An optimal switching approach}. 
Quantitative Finance 10(4) (2010), 359--374.

\bibitem{CT08} R. Carmona, N. Touzi. \textit{Optimal multiple stopping and valuation of swing options}.
Mathematical Finance 18(2) (2008), 239--268.

\bibitem{CV08} A.~Cartea, P.~Villaplana. \textit{Spot price
modeling and the valuation of electricity forward contracts: The
role of demand and capacity}. Journal of Banking and Finance, 32
(12) (2008), 2502--2519.

\bibitem{ChenForsyth06} Z.~Chen, P.~A.~Forsyth. \textit{A
semi-lagrangian approach for natural gas storage, valuation and
optimal operation}. SIAM Journal on Scientific Computing 30 (1)
(2007), 339--368.

\bibitem{DaLioLey} F.~Da Lio, O.~Ley. \textit{Uniqueness results
for second-order Bellman-Isaacs equations under quadratic growth
assumptions and applications}. SIAM J. Control Optim 45 (1)
(2006), 74--106.

\bibitem{Davis06} M.H. Davis. \textit{Optimal hedging with basis risk}. In: From stochastic calculus to mathematical finance (2006), 169--187. Springer Berlin Heidelberg.

\bibitem{ElKR} N.~El Karoui, R.~Rouge. \textit{Pricing via
utility maximization and entropy}. Mathematical Finance 10(2)
(2000), 259--276.

\bibitem{EFRV} E.~Edoli, S.~Fiorenzani, S.~Ravelli, T.~Vargiolu.
\textit{Modeling and valuing make-up clauses in gas swing
contracts}. Energy Economics 35 (2013), 58--73.

\bibitem{ELN} M.~Eriksson, J.~Lempa, T.~K.~Nilssen. \textit{Swing options in commodity markets: a multidimensional L\'evy diffusion model}. Mathematical Methods of Operations Research 79 (1) (2014), 31--67.

\bibitem{Felix12} B.~Felix. \textit{Gas Storage valuation: a
comparative study}. EWL Working Paper N. 01/2012, University of
Duisburg-Essen.

\bibitem{Fiorenz} S.~Fiorenzani. \textit{Pricing illiquidity in
energy markets}. Energy Risk 65 (2006), 65--71.

\bibitem{fleming_rishel} W. Fleming and R. W. Rishel.
\emph{Deterministic and stochastic optimal control}. Springer
Verlag, 1975.

\bibitem{FleSon} W.~Fleming and H.~M. Soner. {\em Controlled
Markov Processes and Viscosity Solutions}. Springer, 1993.

\bibitem{Frittelli00} M.~Frittelli. \textit{The minimal entropy
martingale measure and the valuation problem in incomplete
markets}. Math. Finance 10 (2000), 39--52.

\bibitem{HeLaRusso2013} P. Henaff, I. Laachir, F. Russo. \textit{Gas storage valuation and hedging.
A quantification of the model risk}, 2013, preprint. Available online at  {\tt
http://arxiv.org/abs/1312.3789}.

\bibitem{Henderson02} V. Henderson. \textit{Valuation of claims on nontraded assets using utility maximization}. Mathematical Finance, 12 (2002), 351--373.

\bibitem{HendersonHobson09} V.~Henderson, D.~Hobson.
\textit{Utility Indifference Pricing - An Overview}. Chapter 2 of
{\em Indifference Pricing: Theory and Applications}, ed. R.
Carmona, Princeton University Press, 2009.

\bibitem{LS1} R. Liptser, A. N. Shiryaev. {\em Statistics of Random Processes: I. General Theory}. Vol. 5. Springer Science \& Business Media, 2013.

\bibitem{Ludkovski08} M. Ludkovski. \textit{Financial Hedging of
Operational Flexibility}. International Journal of Theoretical
and Applied Finance, 11(8) (2008), 799--839.

\bibitem{MU} A. Mijatovi\'c, M. Urusov. \textit{On the martingale property of certain local martingales}. Probability Theory and Related Fields 152.1-2 (2012), 1--30.

\bibitem{mnif07} M. Mnif. \textit{Portfolio optimization with stochastic volatilities and constraints: An application in high dimension}. Applied Mathematics and Optimization 56(2) (2007), 243--264.

\bibitem{Monoyios04} M. Monoyios. \textit{Performance of utility-based strategies for hedging basis risk}. Quantitative Finance 4, no. 3 (2004), 245--255.

\bibitem{MusielaZarip04} M.~Musiela, T.~Zariphopoulou. \textit{An
example of indifference prices under exponential preferences}.
Finance and Stochastics 8 (2004), 229--239.

\bibitem{PaBaBo09} G. Pag\`es, O. Bardou, S. Bouthemy. \textit{
Optimal quantization for the pricing of swing options}. Applied 
Mathematical Finance, 16(2) (2009), 183--217.

\bibitem{PaBrWil10}  G. Pag\`es, A.-L. Bronstein, B. Wilbertz.
\textit{How to speed up the quantization tree algorithm with an application to swing options}. 
Quantitative Finance, 10(9) (2010), 995--1007.

\bibitem{Pham02} H. Pham. \textit{Smooth solutions to optimal
investment models with stochastic volatilities and portfolio
constraints}. Applied Mathematics and Optimization 46(1) (2002),
55--78.

\bibitem{PhamBook09} H.~Pham. {\em Continuous-time Stochastic
Control and Optimization with Financial Applications}. Springer,
2009.
    
\bibitem{PorchetTouziWarin} A. Porchet, N. Touzi and X. Warin.
\textit{Valuation of power plants by utility indifference and
numerical computation}. Mathematical Methods of Operations
research 70(1) (2009), 47--75.

\bibitem{OZ03} A.~Oberman, T.~Zariphopoulou. \textit{Pricing
early exercise contracts in incomplete markets}. Computational
Management Science 1 (1) (2003), 75--107.

\bibitem{Rock} R.~T.~Rockafellar. {\em Convex Analysis}.
Princeton University Press, Princeton, NJ 1970.

\bibitem{RW2} L.C.G Rogers, D. Williams. \emph{Diffusions, Markov processes and martingales: Volume 2, It\^o calculus}. Cambridge University Press, 2000.

\bibitem{SS00} E.~Schwartz, J.~E.~Smith. \textit{Short-Term
Variations and Long-Term Dynamics in Commodity Prices}.
Management Science 46 (7) (2000), 893--911.

\bibitem{TDR09} M.~Thompson, M.~Davison, H.~Rasmussen.
\textit{Natural Gas Storage Valuation and Optimization: a real
option application}. Naval Research Logistics (NRL) 56 (3)
(2009), 226--238.

\bibitem{warin} X. Warin. \textit{Gas storage hedging}. In: Numerical
Methods in Finance. Springer Berlin Heidelberg (2012), 421--445.

\end{thebibliography}
\end{document}